\newif\ifabstract

\abstractfalse 
\newif\iffull
\ifabstract \fullfalse \else \fulltrue \fi

\ifabstract
\documentclass{sig-alternate}
\fi

\iffull
\documentclass[11pt]{article}
\usepackage{fullpage}
\fi

\usepackage{amssymb,amsmath}
\iffull
\usepackage{amsthm}
\fi
\usepackage{graphicx}
\usepackage{subfigure}
\usepackage{graphics}
\usepackage{url}
\usepackage{hyperref}

\newtheorem{theorem}{Theorem}[section]
\newtheorem{lemma}[theorem]{Lemma}

\newtheorem{corollary}[theorem]{Corollary}
\newtheorem{definition}[theorem]{Definition}

\newtheorem{remark}{Remark}[section]

\newcommand{\len}{\mathsf{len}}
\newcommand{\genus}{\mathsf{genus}}
\newcommand{\crossingnumber}{\mathsf{cr}}
\newcommand{\edgeplanarization}{\mathsf{edge\text{-}planarization}}
\newcommand{\vertexplanarization}{\mathsf{vertex\text{-}planarization}}
\newcommand{\eg}{\mathsf{eg}}

\newcommand{\dmax}{\Delta}
\newcommand{\OPT}{\mathsf{OPT}}
\newcommand{\framed}{\mathsf{framed}}

\renewcommand{\phi}{\varphi}
\newcommand{\eps}{\varepsilon}

\newcommand{\calS}{{\cal S}}

\newcommand{\poly}{\operatorname{poly}}

\ifabstract
\newcommand{\myparagraph}[1]{\smallskip \noindent {\bf #1 }}
\newcommand{\proofof}{of }
\fi
\iffull
\newcommand{\myparagraph}[1]{\paragraph{#1}}
\newcommand{\proofof}{Proof of }
\fi

\title{Approximation algorithms for Euler genus and related problems\footnote{An extended abstract of this paper appeared in FOCS 2013. This version contains some minor fixes to the previous one.}}

\author{Chandra Chekuri\thanks{Dept.\ of Computer Science, University
of Illinois, Urbana, IL 61801. {\tt chekuri@illinois.edu}.
Supported in part by NSF grant CCF-1016684.} 
\and 
Anastasios Sidiropoulos\thanks{Departments of Computer Science \& Engineering and Mathematics, The Ohio State University, Columbus, OH 43210. {\tt sidiropoulos.1@osu.edu}.
Supported in part by David and Lucille Packard Fellowship, NSF AF award CCF-0915984, and NSF grant CCF-0915519.}
}

\date{\today}

\begin{document}

\pagenumbering{gobble}

\maketitle

\begin{abstract}
  The Euler genus of a graph is a fundamental and well-studied
  parameter in graph theory and topology. Computing it has been shown
  to be NP-hard by Thomassen \cite{Thomassen89,Thomassen93a}, and it
  is known to be fixed-parameter tractable. However, the
  approximability of the Euler genus is wide open. While the existence
  of an $O(1)$-approximation is not ruled out, only an
  $O(\sqrt{n})$-approximation \cite{ChenKK97} is known even in bounded
  degree graphs. In this paper we give a polynomial-time algorithm
  which on input a bounded-degree graph of Euler genus $g$, computes a
  drawing into a surface of Euler genus $g^{O(1)} \cdot \log^{O(1)}
  n$. Combined with the upper bound from \cite{ChenKK97}, our result
  also implies a $O(n^{1/2 - \alpha})$-approximation, for some
  constant $\alpha>0$.
  
  Using our algorithm for approximating the Euler genus as a
  subroutine, we obtain, in a {\em unified} fashion, algorithms with
  approximation ratios of the form $\OPT^{O(1)} \cdot \log^{O(1)} n$
  for several related problems on bounded degree graphs. These include
  the problems of orientable genus, crossing number, and planar edge
  and vertex deletion problems. Our algorithm and proof of correctness
  for the crossing number problem is simpler compared to the long and
  difficult proof in the recent breakthrough by Chuzhoy
  \cite{Chuzhoy11}, while essentially obtaining a qualitatively
  similar result. For planar edge and vertex deletion problems our
  results are the first to obtain a bound of form $\poly(\OPT,\log
  n)$.

  We also highlight some further applications of our results in the
  design of algorithms for graphs with small genus. Many such
  algorithms require that a drawing of the graph is given as part of
  the input. Our results imply that in several interesting cases, we
  can implement such algorithms even when the drawing is unknown.
\end{abstract}

\newpage
\setcounter{page}{1}
\pagenumbering{arabic}

\section{Introduction}
A \emph{drawing} of a graph $G$ into a surface ${\cal S}$ is a mapping
that sends every vertex $v\in V(G)$ into a point $\phi(v)\in {\cal
  S}$, and every edge into a simple curve connecting its endpoints, so
that the images of different edges are allowed to intersect only at
their endpoints. In this paper we deal with closed surfaces (compact
and without boundary). A surface ${\cal S}$ can be orientable or
non-orientable.
The \emph{Euler genus} $\eg({\cal S})$ of a nonorientable surface
${\cal S}$ is defined to be the nonorientable genus of ${\cal S}$.
The Euler genus $\eg({\cal S})$ of an orientable surface ${\cal S}$ is
equal to $2\gamma$, where $\gamma$ is the orientable genus of ${\cal
  S}$.  For a graph $G$, the Euler genus of $G$, denoted by $\eg(G)$,
is defined to be equal to the infimal Euler genus of a surface ${\cal
  S}$, such that $G$ can be drawn into ${\cal S}$. The orientable
genus of a graph $G$, denoted by $\genus(G)$, is the infimal genus of
an orientable surface ${\cal S}$ into which $G$ can be drawn.

Drawings of graphs into various surfaces are of central importance in
graph theory (e.g.~\cite{gross2001topological, mohar2001graphs}),
topology, and mathematics in general (e.g.~\cite{white2001graphs}),
and have been the subject of intensive study. In particular, surface
embedded graphs are an important ingredient in the seminal work of
Robertson and Seymour on graph minors and the proof of Wagner's
conjecture. Surface embedded graphs are also important in computer
science, and engineering, since they can be used to model a wide
variety of natural objects.

We consider two simple-to-state and fundamental optimization problems
in topological graph theory: given a graph $G$, compute $\eg(G)$ and
$\genus(G)$. Thomassen \cite{Thomassen89} showed that computing these
quantities is NP-hard. Deciding whether a graph has Euler genus 0,
i.e.~planarity testing, can be done in linear time by the seminal
result of Hopcroft \& Tarjan \cite{Hopcroft}. Deciding if $\eg(G) \le
g$ is fixed-parameter tractable. In fact, Mohar \cite{Mohar99} gave a
linear time algorithm for this problem, and subsequently a relatively
simple linear-time algorithm was given by Kawarabayashi, Mohar \& Reed
\cite{KawarabayashiMR08}.  The dependence of the running time in the
above mentioned algorithms is at least exponential in $g$.  We note
that, for any fixed $g$, the set of all graphs with genus at most $g$,
denoted by ${\cal G}_g$ is minor-closed.  From the work of Robertson
and Seymour \cite{RobertsonS90b}, ${\cal G}_g$ is characterized as the
class of graphs that exclude as a minor all graphs from a finite
family of graphs ${\cal H}_g$. However, ${\cal H}_g$ is not known
explicitly even for small values of $g$ and $|{\cal H}_g|$
can be very large. We remark that ${\cal H}_1$ is known explicity.

In this paper we consider the case when $g$ is not a fixed constant
and examine the {\em approximability} of $\eg(G)$ and $\genus(G)$.
Perhaps surprisingly, this problem is very poorly understood.  We
briefly describe the known results and illustrate the technical
difficulties. In general, $\eg(G)$ can be as large as $\Omega(n^2)$
where $n$ is the number of nodes of $G$ (e.g.~for the complete graph
$K_n$), and Euler's characteristic implies that any $n$-vertex graph
of Euler genus $g$ has at most $O(n+g)$ edges.  Since any graph can be
drawn into a surface that has one handle for every edge, this
immediately implies an $O(n/g)$-approximation, which is a
$\Theta(n)$-approximation in the worst case. In other words, even
though we currently cannot exclude the existence of an
$O(1)$-approximation, the state of the art only gives a trivial
$O(n)$-approximation.  Using the fact that graphs of small genus have
small balanced vertex-separators, Chen, Kanchi \& Kanevsky
\cite{ChenKK97} obtained a simple $O(\sqrt{n})$-approximation for
graphs of {\em bounded degree} which is currently the best known
approximation ratio. In fact, if we do not assume bounded degree,
nothing better than the trivial $O(n/g)$-approximation is
known. Consider the case of apex graphs which are graphs that contain
a single vertex whose removal makes them planar.  Mohar
\cite{mohar2001face} showed that the genus problem for even these
graphs is NP-hard. He also gave an elegant characterization of the
genus for apex graphs, which can in turn be used to obtain a
$O(1)$-approximation for such graphs. It is worth mentioning that
essentially nothing is known for graphs with a constant number of
apices!  We also remark that by Euler's formula, there is a trivial
$O(1)$-approximation if the average degree is at least $6+\eps$, for
some fixed $\eps>0$. Finally, we mention a recent result by
Makarychev, Nayyeri \& Sidiropoulos \cite{genus_hamiltonian}, who gave
an algorithm that given a Hamiltonian graph $G$ along with a
Hamiltonian path $P$, draws the graph into a surface of Euler genus
$g^{O(1)} \log^{O(1)} n$ where $g$ is the orientable genus of $G$. We
note that their algorithm does not assume bounded degree which is its
strength but assumes Hamiltonicity which is a limitation. Moreover,
the techniques in \cite{genus_hamiltonian} rely heavily on using the
given Hamiltonian path $P$ while our techniques here are based on
treewidth related ideas among several others. 

Our algorithms for approximating genus also give us, in a unified
fashion, algorithms for two related problems on drawing a graph on a
planar surface, namely crossing number and planar edge/vertex
deletion. The guarantees of our algorithms are of the form
$\OPT^{O(1)} \log^{O(1)} n$. These problems have also been
well-studied and have the common feature that the know hardness
results are weak (either NP-Hardness or APX-Hardness) while known
approximation bounds are polynomial in $n$ even in bounded-degree
graphs. In this context we make some remarks on why the bounded-degree
assumption is interesting despite being a limitation in some
ways. First, we can assume that the graph has bounded average degree
since otherwise the lower bound on the instance is very high and it
becomes easy to approximate (see previous comment on genus).  It is
not uncommon in applications such as VLSI design and graph layout to
assume some form of an upper bound on the degree; heuristically
algorithms that work for bounded degree graphs can be extended to
handle the case of graphs that can be made bounded degree by the
removal of a small number of edges.  Second, from a theoretical point
of view, understanding the approximability even when all degrees are
bounded (in fact at most $3$) is non-trivial and there has been very
limited progress over two decades.  It is only very recently that
Chuzhoy, in a breaktrough and technically difficult work, obtained a
bound of the form $\OPT^{O(1)} \log^{O(1)} n$ for crossing number
problem in bounded degree graphs. We now describe our results
formally.

\myparagraph{Our results.}
Our main result is an approximation algorithm for the Euler genus of
bounded degree graphs. More specifically, given a graph $G$ of Euler
genus $g$, our algorithm computes a drawing of $G$ into a surface of
Euler genus $\dmax^{O(1)}g^{O(1)} \log^{O(1)} n$ where $\Delta$ is
the maximum degree.  The algorithm's running time is
polynomial in both $g$ and $n$.  Combined with the simple
$O(n^{1/2})$-approximation from \cite{ChenKK97}, our result
gives a $O(n^{1/2-\alpha})$-approximation for some fixed
constant $\alpha>0$. The following theorem summarizes our main result.


\begin{theorem}[Main result]\label{thm:main}
  There is a polynomial-time algorithm which given a graph $G$ of
  maximum degree $\dmax$, and an integer $g\geq 0$, either outputs a
  drawing of $G$ into a surface of Euler genus $O(\dmax^2 g^{12}
\log^{19/2}n)$, or correctly decides that the Euler genus of $G$ is
  greater than $g$.
\end{theorem}

\begin{remark}
  \label{remark:opt-ratio}
  Kawarabayashi, Mohar and Reed \cite{KawarabayashiMR08} obtain an
  exact algorithm to compute the Euler genus of a given graph in time
  $2^{O(\OPT)} n$ time, which in particular implies a polynomial-time
  algorithm when $\OPT = O(\log n)$; this simplifies and improves a
  previous linear-time algorithm of Mohar~\cite{Mohar99} which had a
  doubly-exponential dependence on $\OPT$. Theorem~\ref{thm:main},
  when combined with the algorithm in \cite{KawarabayashiMR08}, implies a
  polynomial-time algorithm
that given a graph $G$ outputs a drawing
  on a sufrace with Euler genus $O(\dmax^3 \OPT^{O(1)})$.
\end{remark}

We build on our main result to obtain several other non-trivial
results; we describe the outline of the unified methodology in
Section~\ref{sec:reductions}. First, we obtain an
algorithm to approximate $\genus(G)$, the orientable genus of a given
graph $G$, summarized in the theorem below. Note that $\genus(G)$ can
be $\Omega(\sqrt{|V(G)|})$ even when $\eg(G) =
O(1)$~\cite{JGT:JGT3190200305}.

\begin{theorem}[Approximating the orientable
  genus]\label{thm:intro_orientable_genus}
  There exists a polynomial-time algorithm which given a graph $G$ of
  maximum degree $\dmax$, and an integer $g>0$, either correctly
  decides that $\genus(G)>g$, or outputs a drawing of $G$ into a
  surface of orientable genus $O(\dmax^3 g^{14} \log^{19/2} n)$.
\end{theorem}

\noindent
{\em Crossing number.}
In the crossing number problem the input is a graph $G$ which may not be
planar and the goal is to draw it into
the Euclidean plane with as few edge crossings as possible. When we
deal with this problem, we will allow the edges in a graph drawing to
intersect in their interiors.  The point where the interiors of two
edges intersect, is called a \emph{crossing} of the drawing.  We do
not allow the interiors of edges to intersect vertices, and we also
assume that there are no three edges, with their interiors
intersecting at the same point.  The \emph{crossing number} of a graph
$G$, denoted by $\crossingnumber(G)$, is defined to be the smallest
integer $k$, such that $G$ admits a drawing into the plane, with at
most $k$ crossings.

The crossing number problem has also been a difficult problem to
approximate, and the focus has been primarily on bounded degree
graphs.  It is an NP-hard problem but for each fixed $k$ there is a
linear time algorithm to decide if $\crossingnumber(G) \le k$
\cite{KawarabayashiR07}. In a recent breakthrough paper, Chuzhoy
\cite{Chuzhoy11} described an algorithm that given a graph $G$ outputs
a drawing into the plane with $O(\crossingnumber(G)^{10} \poly(\dmax
\log n))$ crossings; as a corollary she obtained the first algorithm
that had an approximation ratio that is sub-linear in $|V(G)|$. The
algorithm and proof in \cite{Chuzhoy11} occupy almost 80 pages. It is
a simple observation that if the crossing number of a graph $G$ is $k$
then $\genus(G) \le k$ since one can add a handle for each edge that
participates in a crossing. We can apply our approximation algorithm
to find a drawing of $G$ into an orientable surface, via
Theorem~\ref{thm:intro_orientable_genus}, of genus $O(\dmax^4 k^{9}
\log^{13/2} n)$. Interestingly, having a drawing on a relatively low
genus surface, allows us to obtain a rather simple algorithm for
crossing number. Our result is summarized below.

\begin{theorem}[Approximating the crossing number]
\label{thm:intro_crossing_number}
  There exists a polynomial-time algorithm which given a graph $G$ of maximum degree $\dmax$,
  and an integer $k\geq 0$, either correctly decides that
  $\crossingnumber(G)>k$, or outputs a drawing of $G$ into the plane
  with at most $O(\dmax^{9} k^{30} \log^{19} n)$ crossings.
\end{theorem}

We note that the dependence on $k$ in our theorem is worse than that
in \cite{Chuzhoy11}. However, we believe that our approach, in
addition to giving a simpler proof, is interesting because it appears
to differ from that in \cite{Chuzhoy11} in going via a somewhat
indirect route through a low genus drawing.
We refer the interested reader to \cite{ChuzhoyMS11,Chuzhoy11,ChimaniH11} for
various pointers to the extensive work on crossing number and related
problems.

\medskip
\noindent
{\em Planar Edge and Vertex Deletion.}
We extend our approach via genus to obtain an approximation algorithm
for the minimum planar edge/vertex deletion problems.  In these
problems we are given a graph $G$ and the goal is to remove the
smallest number of edges/vertices to make it planar.  We denote by
$\edgeplanarization(G)$ the minimum size of such a planarizing set of
edges and similarly by $\vertexplanarization(G)$ for vertices. The
best known approximation for this problems is $O(\sqrt{n \log n})$
due to Tragoudas via the separator algorithms \cite{LeightonR99}, and recently
Chuzhoy \cite{Chuzhoy11} gives an algorithm that outputs a solution of size
$O({\crossingnumber(G)}^5 \poly(\dmax \cdot \log n))$; we observe that
the $\crossingnumber(G)$ could be $\Omega(\sqrt{n})$ even though there
may be a single edge $e$ such that $G-e$ is planar. We obtain the
first non-trivial approximation algorithm for this problem.  Our
result is summarized in the following Theorem.

\begin{theorem}[Approximating the minimum planar edge/vertex deletion]
  There exists a polynomial-time algorithm which given a graph $G$ of
  maximum degree $\dmax$, and an integer $k>0$, either correctly
  decides that $\edgeplanarization(G)>k$, or outputs a set $Y\subseteq
  E(G)$, with $|Y| = O(\dmax^5 k^{15} \log^{19/2} n)$, such that
  $G\setminus Y$ is planar. Similarly, there is a polynomial-time
  algorithm that either correctly decides that $\vertexplanarization(G) > k$
  or outputs a set $X \subset V$ with $|X| = O(\dmax^4 k^{15} \log^{19/2} n)$
  such that $G\setminus X$ is planar.
\end{theorem}

\begin{remark}
  Our approach via genus gives algorithms with ratios
  $O(\dmax^{O(1)}\OPT^{O(1)})$ for crossing number and planar
  edge/vertex deletion. It is useful to note that, unlike
  for genus, crossing number and planar edge/vertex deletion do not
  yet have fixed-parameter-tractable algorithms that have a
  singly-exponential dependence on $\OPT$.
\end{remark}

\paragraph{Further algorithmic applications.}
Our approximation algorithm for Euler genus has further consequences
in the design of algorithms for problems on graphs of small
genus. Most algorithms that take advantage of the fact that a graph
can be drawn on a surface of small genus require a drawing of the
input graph be given as part of the input. If the genus $g = O(\log
n)$ then one can use existing exact algorithms that run in $2^{O(g)}
\text{poly}(n)$ time to obtain a drawing. Our result implies that we
can obtain a drawing even when $g = \Omega(\log n)$, that while not
being optimal, nevertheless yields interesting results. A concrete
example of this application is the following.  Recently, Erickson and
Sidiropoulos \cite{erickson_sid_atsp} have obtained a $O(\log g / \log
\log g)$-approximation for Asymmetric TSP on graph of Euler genus $g$,
when a drawing of the graph is given as part of the input; this
improves the bounds of Oveis-Gharan and Saberi
\cite{DBLP:conf/soda/GharanS11} who gave an $O(\sqrt{g} \log g)$-approximation 
and also required the drawing as an input.
Our result implies the following corollary: There exists a
polynomial-time $O(\log g / \log \log g)$-approximation for ATSP on
bounded-degree graphs of genus $g$, even when a drawing of the graph
is not given as part of the input\footnote{More precisely, there
  exists a polynomial-time algorithm which given a bounded-degree
  graph $G$ (the instance of ATSP), and an integer $g$, either
  correctly decides that $\eg(G)>g$, or it outputs a $O(\log g / \log
  \log g)$-approximate TSP tour in $G$.}.


\medskip
\noindent
The proof of Theorem~\ref{thm:main} is somewhat technical and uses
several ingredients. To aid the reader we first give an overview of 
the algorithmic ideas and highlight the ingredients that are needed. We
assume that the reader is familiar with the notion of the treewidth
of a graph. Section~\ref{sec:reductions} highlights the high-level idea
that allows us to leverage an algorithm for Euler genus for the other
problems considered in the paper.

\subsection{Overview of the algorithm for Euler genus}
\label{sec:overview}
It is convenient to work with a promise version of the problem
where we assume that $\eg(G)$ is at most a given number $g$.
This allows us to assume certain properties that $G$ needs to satisfy.
Our algorithm may find that $G$ does not satisfy such a property in which
case it obtains a certificate that $\eg(G) > g$.

\myparagraph{An idea from exact algorithms.}
Our algorithm is inspired by fixed-parameter algorithms that run in
polynomial time for any fixed genus
\cite{RobertsonS90b,Mohar99,KawarabayashiMR08}.  It is instructive to
briefly describe how such algorithms work.  Let $G$ denote the input
graph, and suppose we want to find a drawing into a surface of Euler
genus $g$, if one exists.  If $G$ happens to have bounded treewidth,
say $f(g)$ for some function $f$, then one can compute its Euler genus
exactly via a dynamic program, in time roughly $2^{O(f(g))} n^{O(1)}$.
If on the other hand $G$ has treewidth larger than $f(g)$, by choosing
$f$ to be sufficiently large, a theorem of Robertson, and Seymour
\cite{DBLP:journals/jct/RobertsonS86,DBLP:journals/jct/RobertsonS03a}, asserts that $G$ contains a \emph{flat} $((2g+1)\times
(2g+1))$-grid minor $H$.  Here, being flat means that the graph $H$
admits a planar drawing, such that all edges leaving $H$ are incident
to the outer face.  The central vertex $v$ of such a grid can be shown
to be \emph{irrelevant}, i.e.~such that $G$ admits a drawing into a
surface of Euler genus $g$, if and only if $G - v$ does.  Therefore,
we can simply remove $v$ from $G$, and recurse on the remaining graph.
We continue removing irrelevant vertices in this fashion, until the
treewidth becomes at most $f(g)$.  We call the resulting low-treewidth
graph a \emph{skeleton} of $G$ (see Figure \ref{fig:exact_skeleton}).
After drawing the skeleton, we can extend the drawing to all the
removed irrelevant vertices.

\begin{figure}
\begin{center}
\iffull
\scalebox{0.9}{\includegraphics{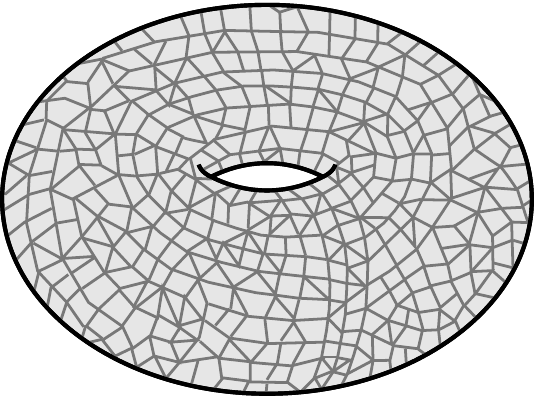}}
\hspace{1cm}
\scalebox{0.9}{\includegraphics{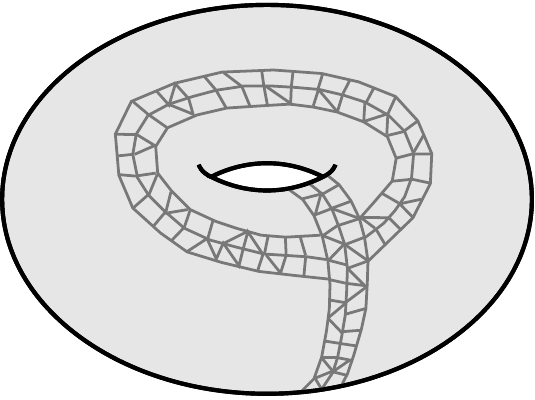}}
\fi
\ifabstract
\scalebox{0.6}{\includegraphics{figs/torus1}}
\hspace{1cm}
\scalebox{0.6}{\includegraphics{figs/torus2}}
\fi
\caption{A high-treewidth graph drawn into the torus (left), and the low-treewidth skeleton obtained after removing irrelevant vertices (right).\label{fig:exact_skeleton}}
\end{center}
\end{figure}

\myparagraph{Challenges when $g$ is not  a fixed constant.}
Our algorithm is based on modifying the above approach, so that it
works in the approximate setting when $g$ is part of the input.  We
now briefly describe the main challenges towards this goal.  Let us
begin with considering the case of a bounded-degree graph $G$ of small
treewidth, say at most $g^{O(1)}$, where $g$ is the Euler genus of
$G$.  By repeatedly cutting along balanced separators, we can compute
in polynomial time a set of at most $\dmax^{O(1)}g^{O(1)} \log^{O(1)} n$ edges
$E^*\subset E(G)$, such that $G\setminus E^*$ is planar.  By
introducing one handle for every edge in $E^*$, we get a drawing of
$G$, into a surface of Euler genus (in fact orientable genus)
$\dmax^{O(1)}g^{O(1)} \log^{O(1)} n$ .

Let us now consider the general case when treewidth of the graph $G$
is larger than $g^{c}$ for some sufficiently large constant $c$. Let
us assume for now that we can again find an irrelevant vertex in $G$.
It might seem at first that we are done, by proceeding as in the exact
case and recursing on the reduced instance.  However, this is the
critical point where things break down in the approximate setting. 
Suppose that we remove a set
$U\subset V(G)$ of irrelevant vertices, such that the skeleton
$G\setminus U$ has treewidth $g^{O(1)}$. We know that the skeleton
can be embedded with genus $g$ iff $G$ can. However, we only have
an approximate algorithm for handling a low-treewidth graph.
Using such an algorithm, we can compute a drawing $\phi$
of $G\setminus U$ into a surface of Euler genus $\dmax^{O(1)} g^{O(1)} \log^{O(1)}
n$.  Unfortunately, now we are stuck!  Since the drawing $\phi$ is not
into a surface of Euler genus $g$, there might be no way of
extending $\phi$ to $U$.

\begin{figure}
\begin{center}
\iffull
  \subfigure[A collection of patches (in bold).]{
    \scalebox{0.9}{\includegraphics{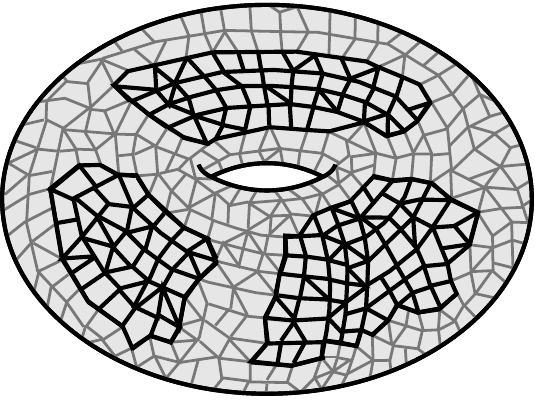}}
    \label{fig:intro1}
  }
  \subfigure[The skeleton obtained after removing the interiors of all patches.]{
    \scalebox{0.9}{\includegraphics{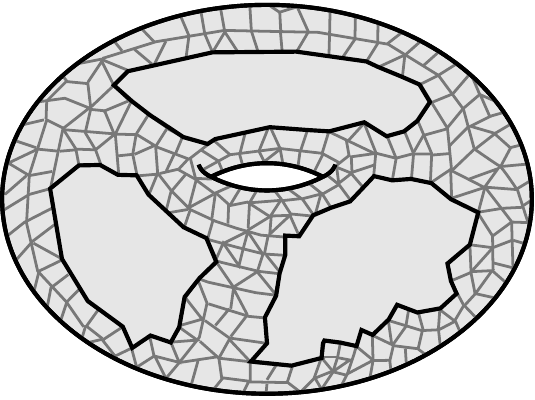}}
    \label{fig:intro2}
  }
  \subfigure[The graph $G''$ obtained after attaching a width-3 cylinder along every boundary cycle.]{
    \scalebox{0.9}{\includegraphics{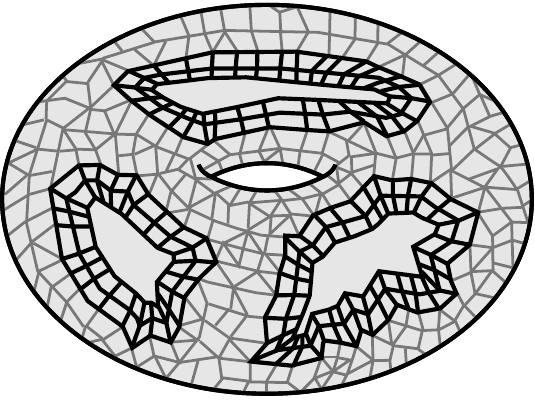}}
    \label{fig:intro3}
  }
\fi
\ifabstract
  \subfigure[A collection of patches (in bold).]{
    \scalebox{0.6}{\includegraphics{figs/intro1}}
    \label{fig:intro1}
  }
  \subfigure[The skeleton obtained after removing the interiors of all patches.]{
    \scalebox{0.6}{\includegraphics{figs/intro2}}
    \label{fig:intro2}
  }
  \subfigure[The graph $G''$ obtained after attaching a width-3 cylinder along every boundary cycle.]{
    \scalebox{0.6}{\includegraphics{figs/intro3}}
    \label{fig:intro3}
  }
\fi
\end{center}
\caption{Constructing a skeleton\iffull by removing patches\fi.
\label{fig:intro}}
\end{figure}

\myparagraph{Ensuring extendability.}
We overcome the above issue by carefully computing irrelevant parts,
that have some extra structure.  This structure guarantees that the
resulting approximate drawing of the skeleton can be extended to the
whole $G$, by introducing only a small number of additional handles.
To that end, we define a structure that we call a \emph{patch}.  A
patch is simply a subgraph $X\subset G$, together with a cycle $C$,
which we can think of as its ``boundary''.  We also think of
$X\setminus C$ as the ``interior'' of the patch.  Our goal is to
compute patches $X_1,\ldots,X_k$, satisfying the following two
conditions:
\begin{description}
\item{(C1)}
After removing the interiors of all patches, the resulting skeleton has treewidth at most $g^{O(1)}$.
\item{(C2)} There exists a drawing $\phi_{\OPT}$ of $G$ into a surface
  ${\cal S}$ of genus $\eg(G)$, such that every patch $X_i$ is drawn
  inside a disk ${\cal D}_i$, with its boundary being mapped to the
  boundary of ${\cal D}_i$.  Moreover, the disks ${\cal D}_i$ have
  pairwise disjoint interiors, and there is no part of $G$ drawn
  inside each ${\cal D}_i$, other than $X_i$ (see Figure
  \ref{fig:intro1}).  We remark that we do not explicitly know
  $\phi_{OPT}$, but we can nevertheless guarantee its existence.
\end{description}

Let us suppose for now that we can compute such a skeleton, with a
corresponding collection of patches $X_1,\ldots,X_k$.  Let $C_i$ be
the boundary cycle of each $X_i$.  Let $G'$ be the skeleton $G
\setminus \bigcup_{i=1}^k (X_i \setminus C_i)$ (see Figure
\ref{fig:intro2}).  Let us now revisit the algorithm for low-treewidth
graphs: We repeatedly remove balanced separators of size $g^{O(1)}
\log^{O(1)}n$, until all connected components become planar.  After
removing a set $E'$ of at most $g^{O(1)} \log^{O(1)}n$ edges, we end up
with a planar graph $H' = G' \setminus E'$.  Fix a planar drawing
$\phi'$ of $H'$.  We would like to extend $\phi'$ to a low-genus
drawing of the whole $G$.  To that end, ideally, we would like every
cycle $C_i$ to bound a face in $\phi'$.  There are two things that can
go wrong:
\begin{description}
\item{(P1)}
A cycle might be broken into several different paths.
\item{(P2)}
A maximal segment $P$ of a cycle $C_i$ in $H'$ might not be ``one-sided''.
That is, there might be no face of $\phi'$ containing $P$ as a subpath.
\end{description}

Problem (P1) can be easily addressed: If a cycle gets broken into $t$
pieces, then we can ``fix'' this by adding at most $t$ extra handles.
Since we remove only a small number of edges, and every edge can break
at most two cycles, it follows that we only need to add a small number
of new handles because of (P1).

Overcoming problem (P2) is somewhat more difficult: Intuitively, while
computing the drawing of the skeleton $G'$, we modify $G'$ by
attaching a cylinder of width 3 on each cycle $C_i$ (see Figure
\ref{fig:intro3}). This ensures that in the resulting planar drawing
of $H'$, each segment of every cycle is one-sided. In reality, things
are more complicated, but this is the high-level idea.  After
computing a planar drawing $\phi'$ as above, where every segment of a
cycle is one-sided, we can extend $\phi'$ to a low-genus drawing of
$G$.

\myparagraph{Computing the skeleton.}
The missing ingredient is an algorithm to compute the skeleton
satisfying the conditions described above.  The challenging part is to
satisfy condition (C2).  One difficulty is that we can only compute
patches iteratively.  Hence, if we compute the patches naively, it is
possible that a patch can ``interfere'' with previous patches.  We
avoid this by showing that every new patch, either is
interior-disjoint from all previous ones, or it contains some of them
in its interior.  In the latter case, we can simply merge all internal
patches with the current one.  This is the technical part of the
paper.  Our proof uses several tools from the theory of graph minors,
and topological graph theory, such as the grid-minor/treewidth
duality, Whitney flips, and results on the so-called
\emph{planarly-nested sequences} \cite{Mohar_local_planarity}.

\subsection{Orientable genus, Crossing number and 
Planar edge/vertex deletion}
\label{sec:reductions}
Our algorithms for $\genus(G)$, $\crossingnumber(G)$,
$\edgeplanarization(G)$ and $\vertexplanarization(G)$ rely on the
algorithm for $\eg(G)$. Interestingly having a drawing (even if it is
into a non-orientable surface) helps via the following conceptually
simple methodology. First we consider the problem of computing
$\genus(G)$.  Suppose we have a drawing $\phi$ of $G$ into a surface
$\calS$ whose Euler genus is $g^{O(1)}\log^{O(1)} n$ where $g$ is
$\genus(G)$. Note that $\eg(G) \le \genus(G)$ and hence $\eg(G)$
provides a lower bound for $\genus(G)$. We can efficiently check if
$\calS$ is orientable or non-orientable.  If $\calS$ is orientable
then we are done. Suppose not. Then we compute $\rho$, the {\em
  representativity} (equivalently facewidth) of the drawing $\phi$
which captures how ``densely'' $G$ is embedded in the surface --- see
Section~\ref{sec:orientable_genus} for a formal definition. If $\rho$
is ``small'' relative to $g$ we can cut a small number of edges along
a non-contractible cycle and reduce the genus of the surface.
We repeat this process until we obtain a drawing into a surface ${\cal
  S}'$, such that either ${\cal S}'$ is orientable, or ${\cal S}'$ is
nonorientable, and the representativity is ``large''.  If $\calS'$ is
orientable then we can add handles for all the edges cut along the way
and obtain a drawing of the original graph into an orientable
surface. The interesting case is when $\calS'$ is non-orientable and
has high representativity. In this case we can show via results in
\cite{DBLP:journals/jct/BrunetMR96} that $G$ has a large M\"{o}bius
grid minor that certifies that $\genus(G) > g$. 

A similar approach works for $\crossingnumber(G)$ and
$\edgeplanarization(G)$. It is an easy observation that for each of
these problems we have $\OPT \ge \genus(G)$ where $\OPT$ is the
optimum value for the corresponding problem.  Thus we can use our
algorithm for $\genus(G)$ to first obtain an embedding into an
orientable surface of genus comparable to $\OPT$. We once again use
the idea of representativity. Either we can iteratively keep cutting
along short non-contractible cycles to reduce the genus by at least
one in each step and obtain a planar graph, or we get stuck with an
embedding on a non-planar surface with large representativity. In the
latter case we find a certificate that $\OPT$ is large. In the former
case we need to handle the small number of edges removed to obtain the
planar graph. There is nothing to do for planar edge deletion since
they are part of the output. For crossing number we can add these
edges to the planar graph without incurring too many crossings via the
results in \cite{ChuzhoyMS11,ChimaniH11}.

\smallskip
\noindent {\em Discussion:} 
One could argue that the main reason for the difficulty in
approximating graph drawing problems is to get a suitable lower bound
on the optimum value. Previous algorithms were based on divide and
conquer based approach \cite{BhattL84,LeightonR99,ChenKK97}. However,
this approach incurs an additive term that depends on the size of a
graph and therefore one only obtains a polynomial-factor
approximation. On the other hand the problems are fixed parameter
tractable so when $\OPT$ is quite small, one can obtain an exact
algorithm.  Chuzhoy's algorithm for crossing number, and our results,
address the intermediate regime when $\OPT$ is not too small but is
not so large that an additive term that depends on $n$ can be ignored.
Chuzhoy's algorithm and analysis are technically very involved but in
essense her algorithm finds large rigid substructures in the given
graph (via well-linked sets and grid minors) that have to be
necessarily planar in any drawing with crossing number at most
$\crossingnumber(G)$.  Our algorithms for crossing number and planar
edge/vertex deletion, are indirect in that they are based on algoritms
for genus. Consequently, the bounds we obtain are quantitatively somewhat weaker than those of Chuzhoy for crossing number. However, our
algorithm offers a different perspective and approach which we
believe is more transparent and easier to understand. We hope this
will lead to a better understanding of the problem complexity and to
improved algorithms.

\iffull
\subsection{Organization}
The rest of the paper is organized as follows.  In Section
\ref{sec:notation} we introduce some basic definitions.  In Section
\ref{sec:normalization} we give a procedure for simplifying the input
graph. In Section \ref{sec:algorithm} we present our algorithm for
approximating the Euler genus, assuming an algorithm for computing the
skeleton.  In Section \ref{sec:skeleton} we give the algorithm for
computing the skeleton, assuming an algorithm for computing patches.
The computation of patches uses as a subroutine an algorithm for
computing flat grid minors, which is described in Section
\ref{sec:flat_grid_minors}.  The actual algorithm for computing
patches is given in Section \ref{sec:universal_patch}.  Finally, our
approximation algorithms for orientable genus, crossing number, and
planar edge/vertex deletion are given is Sections
\ref{sec:orientable_genus}, \ref{sec:crossing_number} and \ref{sec:planarization}
respectively.

\fi \ifabstract \medskip
\noindent
{\em Organization:} 
Due to space constraints several proofs and our algorithms for
orientable genus and crossing number are omitted.  We refer the reader
to the full version that has been made available along with this version.
Section \ref{sec:notation} has some basics and a
procedure to simplify the input graph. Section \ref{sec:algorithm}
contains the formal description of our algorithm for Euler genus;
it assumes an algorithm for computing the skeleton which is given in
Section \ref{sec:skeleton}. The skeleton computation is based on
computing patches which is described in Section \ref{sec:universal_patch}. 
Flat grid minors are used as subroutine in patch computation and
a brief outline is given in Section \ref{sec:flat_grid_minors}. 
\fi
\section{Preliminaries}\label{sec:notation}


For an orientable surface ${\cal S}$, let $\genus({\cal S})$ denote
its orientable genus.  Similarly, for a graph $G$, let $\genus(G)$
denote its orientable genus. For a graph $G$, and for $X, Y\subseteq V(G)$, 
we use $E(X,Y)$ to denote the set of edges with one end point in $X$
and the other in $Y$. For $X \subseteq V(G)$ we use $N_G(X)$ to denote
the neighbors of $X$, namely the set of vertices in $V(G) \setminus X$
that have an edge to some vertex in $X$.

A graph $H$ is a minor of a graph $G$ if it is obtained from $G$ by
a sequence of edge deletions, edge contractions, and deletions of isolated vertices.

\begin{definition}[Minor mapping]
  Let $G$ be a graph, and let $H$ be a minor of $G$.  Then there
  exists a function $\sigma : V(H) \to 2^{V(G)}$, satisfying the
  following conditions:
\begin{description}
\item{(1)} For every $v\in V(H)$, $\sigma(v)$ induces a connected
  subgraph in $G$.
\item{(2)} For any $u\neq v \in V(H)$, we have $\sigma(u) \cap
  \sigma(v) = \emptyset$.
\item{(3)} For any $\{u,v\} \in E(H)$, there exist $u' \in \sigma(u)$,
  and $v' \in \sigma(v)$, such that $\{u', v'\} \in E(G)$.
\end{description}
We refer to $\sigma$ as a \emph{minor mapping} (for $H$).  For a set
$U\subset V(H)$, we will use the notation $\sigma(U) = \bigcup_{v \in
  U} \{\sigma(v)\}$.
\end{definition}

\begin{definition}[Grids and cylinders]
  Let $r\geq 1$, $k\geq 3$.  We define the \emph{$(r\times
    k)$-cylinder} to be the Cartesian product of the $r$-path $P$,
  with the $k$-cycle $C$.  We fix an endpoint $v$ of $P$, and let $u$
  be the other endpoint.  We refer to the copy of the $k$-cycle
  $\{v\}\times C$, as the \emph{top}, and to $\{u\}\times C$, as the
  \emph{bottom} (of the cylinder).

  Similarly, for $s\geq 1$, $t\geq 1$, the \emph{$(s \times t)$-grid}
  is the Cartesian product of the $s$-path $P$, with the $t$-path $Q$.
  We fix an endpoint $v$ of $P$, and let $u$ be the other endpoint.
  We refer to the copy of the $t$-path $\{v\}\times Q$, as the
  \emph{top}, and to $\{u\}\times Q$, as the \emph{bottom} (of the
  grid).
\iffull  See Figure \ref{fig:grid_cylinder}. \fi
\end{definition}

\iffull
\begin{figure}
\begin{center}
\iffull
\scalebox{0.8}{\includegraphics{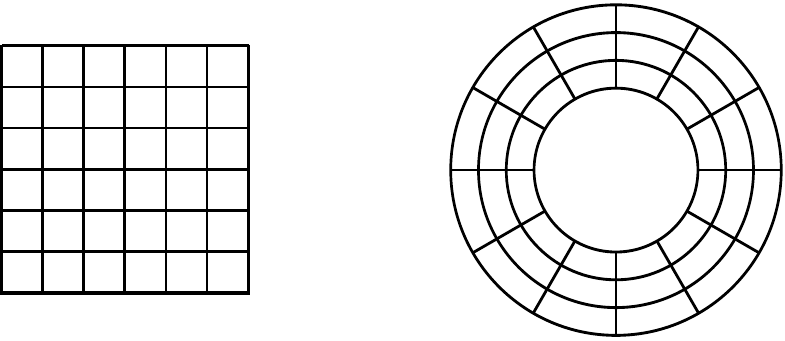}}
\fi
\ifabstract
\scalebox{0.7}{\includegraphics{figs/grid_cylinder}}
\fi
\caption{The $(7\times 7)$-grid  and the $(4\times 12)$-cylinder.\label{fig:grid_cylinder}}
\end{center}
\end{figure}
\fi

We will make use of the following result of Feige et al.~for computing
balanced vertex-separators.
\begin{theorem}[Feige et
  al.~\cite{DBLP:journals/siamcomp/FeigeHL08}]\label{thm:approx_treewidth}
  There exists a polynomial time $O(\sqrt{\log n})$-pseudo
  approximation for balanced vertex separators.  Moreover, given a
  graph $G$ of treewidth $t$, we can compute in polynomial time a tree
  decomposition of $G$ of width $O(t \sqrt{\log t})$.
\end{theorem}

\iffull
\section{Graph normalization}\label{sec:normalization}
\fi
\ifabstract
\subsection{Graph normalization}\label{sec:normalization}
\fi

Before we begin with the description of our algorithm, we give a
procedure for simplifying the input graph.  Throughout the proof of
the main result, we will need to compute and maintain structures that
satisfy certain properties in \emph{any} optimal drawing.  In order to
achieve this, we need to enforce a certain type of ``local rigidity''
of drawings.  To that end, it suffices to ensure that there are no
planar components that can ``flip'' along a small vertex separator.
The following is a formal definition of precisely this situation.

\begin{definition}[Freedom]
Let $G$ be a graph, and let $H\subseteq G$ be a vertex-induced subgraph of $G$.
We say that $H$ is \emph{free} (in $G$) if it satisfies the following conditions:
\begin{description}
\item{(1)}
There exist at most two vertices in $V(H)$, called \emph{portals}, with neighbors in $V(G) \setminus V(H)$.
\item{(2)}
If $H$ has two portals $t,t'$, then $H$ is not a path between $t$ and $t'$.
\item{(3)} There exists a planar drawing of $H$ such that all portals lie in the boundary of the outer face.
\end{description}
If there exists only one portal, then we say that $H$ is a \emph{petal}, and
if there exist two portals, then we say that it is a \emph{clump}.
\end{definition}

\begin{definition}[Normalized graph]
  We say that a graph $G$ is \emph{normalized} if there exists no free
  subgraph $H\subseteq G$.
\end{definition}

The following lemma allows us to restrict our attention to normalized
graphs.  A similar statement is proven in
\cite{Mohar_local_planarity}.

\begin{lemma}
  Given a graph $G$ of maximum degree $\dmax$, we can compute in
  polynomial time a graph $G'$ of maximum degree at most $\dmax$,
  satisfying the following conditions:
\begin{description}
\item{(1)}
The graph $G'$ is normalized.
\item{(2)} Given a drawing of $G'$ into a surface ${\cal S}$, we can
  compute in polynomial time a drawing of $G$ into ${\cal S}$.
\end{description}
\end{lemma}
\iffull
\begin{proof}
  If $G$ is normalized, then we can set $G'=G$.  Otherwise, we start
  by computing an integer $t\geq 0$, and a sequence of graphs
  $G=G_0,\ldots,G_t$.  The graph $G_t$ will be the desired normalized
  graph $G'$.  Suppose we have computed $G_i$.  If $G_i$ is
  normalized, then we set $t=i$.  Otherwise, we find $X_i\subseteq
  V(G)$, with $|X_i|\leq 2$, and such that some connected component
  $C_i$ of $G_i\setminus X_i$ is planar.  If $|X_i|=1$, then we set
  $G_{i+1}=G_i \setminus (V(C_i)\setminus X_i)$.  Otherwise, if
  $|X_i|=2$, we remove all vertices in $V(G_i)\setminus X_i$, and we
  add an edge $e_i$ between the two vertices in $X_i$.  Since
  $|V(G_{i+1})| < |V(G_i)|$, the above process terminates in
  polynomial time, with a normalized graph $G_t=G'$.

  It suffices to show that given a drawing $\phi_{i+1}$ of $G_{i+1}$
  into a surface ${\cal S}$, we can compute in polynomial time a
  drawing $\phi_i$ of $G_i$ into the same surface.  Let $H_i =
  G_i[X_i]$, and let $T_i = X_i \cap V(G_i)$ be the set of portals of
  $H_i$.  Since $H_i$ is free, it admits a planar drawing $\psi_i$, in
  which all its portals lie in the boundary of the outer face.
  Therefore, there exists a drawing $\psi_i'$ of $H_i$ into a disk
  ${\cal D}_i$, such that $\psi_i'(H_i) \cap \partial {\cal D}_i =
  \psi_i'(T_i)$.  If $H_i$ is a petal, then $T_i = \{t\}$.  There
  exists a disk ${\cal R}_i \subset {\cal S}$, intersecting
  $\phi_{i+1}(G_{i+1})$ only on $\phi_{i+1}(t)$.  Embedding the disk
  ${\cal D}_i$ onto ${\cal R}_i$ results into the desired drawing for
  $G_i$.  Similarly, if $H_i$ is a clump, then there exists an edge
  $e_i\{t_1,t_2\}\in E(G_{i+1})$, where $t_1,t_2$ are the two portals
  of $H_i$.  There exists a disk ${\cal R}_i' \subset {\cal S}$ such
  that $\phi_{i+1}(G_{i+1}) \cap {\cal R}_i' = \phi_{i+1}(e_i)$.
  Embedding the disk ${\cal D}_i$ onto ${\cal R}_i'$ results into the
  desired drawing for $G_i$.  This concludes the proof.
\end{proof}
\fi

\section{The algorithm}\label{sec:algorithm}

We begin by formally defining the notion of a patch, which we alluded to in Section \ref{sec:overview}.

\begin{definition}[Patch]
Let $G$ be a graph.
Let $X\subseteq G$ be a subgraph, and let $C\subsetneq X$ be a cycle.
Then, we say that the ordered pair $(X,C)$ is a \emph{patch} (of $G$).
\end{definition}

Note that the above definition of a patch is completely combinatorial, i.e.~it is completely independent from drawings of the graph $G$.
We will often refer to a patch, w.r.t. a specific drawing.
This is captured in the following definition.

\begin{definition}[$\phi$-Patch]
Let $G$ be a graph, and let $(X,C)$ be a patch of $G$.
Let $\phi$ be a drawing of $G$ into a surface ${\cal S}$.
We say that $(X, C)$ is a \emph{$\phi$-patch} (of $G$), if there exists a disk ${\cal D} \subset {\cal S}$, satisfying the following conditions:
\begin{description}
\item{(1)}
$\partial {\cal D} = \phi(C)$.
\item{(2)}
$\phi(G) \cap {\cal D} = \phi(X)$.
\end{description}
\end{definition}

The following definition captures the notion of a pair of ``interfering'' patches.

\begin{definition}[Overlapping patches]
Let $G$ be a graph, and let $(X_1,C_1)$, $(X_2,C_2)$ be patches of $G$.
We say that $(X_1,C_1)$, and $(X_2,C_2)$ are \emph{overlapping} if either $(X_1\setminus C_1) \cap X_2 \neq \emptyset$, or $(X_2\setminus C_2) \cap X_1 \neq \emptyset$.
In particular, if $(X_1,C_1)$, and $(X_2,C_2)$ are non-overlapping, then this definition implies $X_1\cap X_2 = C_1 \cap C_2$.
\end{definition}

Our general goal will be to compute patches that do not interfere precisely in the above sense.
We now have all the notation in place, to state the main result for computing a skeleton of the input graph.

\begin{lemma}[Computing a skeleton]\label{lem:skeleton}
There exists a polynomial-time algorithm which given a graph $G$ of treewidth $t\geq 1$, and maximum degree $\dmax$, and an integer $g>0$, either correctly decides that $\eg(G) > g$, or outputs a collection of pairwise non-overlapping patches $(X_1,C_1),\ldots,(X_r,C_r)$ of $G$, so that the following conditions are satisfied:
\begin{description}
\item{(1)}
If $\eg(G)\leq g$, then there exists a drawing $\phi$ of $G$ into a surface of Euler genus $g$, such that for any $i \in \{1,\ldots, r\}$, $(X_i, C_i)$ is a $\phi$-patch.
We emphasize the fact that $\phi$ is not explicitly computed by the algorithm.
\item{(2)}
The graph $G\setminus \left(\bigcup_{i=1}^r (X_i \setminus C_i) \right)$ has treewidth at most $O(\dmax g^{11} \log^{8} n)$.
\end{description}
\end{lemma}

Lemma \ref{lem:skeleton} is the main technical part of the paper.  In
the interest of clarity, we postpone its proof to later sections, and
we instead show now how it can be used to obtain our approximation
algorithm for Euler genus.

Before we describe the actual algorithm, we need to define a local
``framing'' operation, which we use to modify the skeleton.
Intuitively, this is needed to ensure that when computing a drawing
for the skeleton, the boundaries of the patches are drawn in a
``nearly one-sided'' fashion.  This ``near one-sidedness'' in turn
will allow us to extend the drawing of the skeleton, to a drawing of
the whole graph. Note that framing is a combinatorial operation and
does not rely on a drawing.

\begin{definition}[Graph framing]
  Let $G$ be a graph, and let ${\cal C}=C_1,\ldots,C_k \subseteq G$ be
  a collection of cycles.  Let $G'$ be the graph obtained from $G$ by
  taking for every $i\in \{1,\ldots,k\}$, a copy $K_i$ of the
  $(3\times |V(C_i)|)$-cylinder, and identifying the top of $K_i$ with
  $C_i$.  We refer to $G'$ as the \emph{${\cal C}$-framing of $G$}
  (see Figure \ref{fig:framing_G}).

  More generally, we define the framing operation for subgraphs.  Let
  $H \subseteq G$ be a subgraph of $G$.  We define a graph $H'$ as
  follows.  Consider some $C_i \in {\cal C}$.  If $C_i\subseteq H$, then
  we take a copy of the $(3\times |V(C_i)|)$-cylinder, and we identify
  its top with $C_i$.  If $C_i \not\subseteq H$, then let
  $P_1,\ldots,P_a$ be the set of maximal subpaths of $C_i$ that are
  contained in $H$.  For every such $P_j$, we take a copy of the
  $(3\times |V(P_j)|)$-grid, and we identify its top with $P_j$.  We
  repeat this process for all $C_i\in {\cal C}$, and we define $H'$ to
  be the resulting graph.  We refer to $H'$ as the \emph{${\cal
      C}$-framing of $H$} (see Figure \ref{fig:framing_H}).  The
  reader can check that the definition of the ${\cal C}$-framing of
  $H$ agrees with the one given above, when $H=G$.
\end{definition}

\begin{figure}
\begin{center}
  \subfigure[$\{C_1,C_2\}$-Framing of a graph.]{
    \scalebox{0.72}{\includegraphics{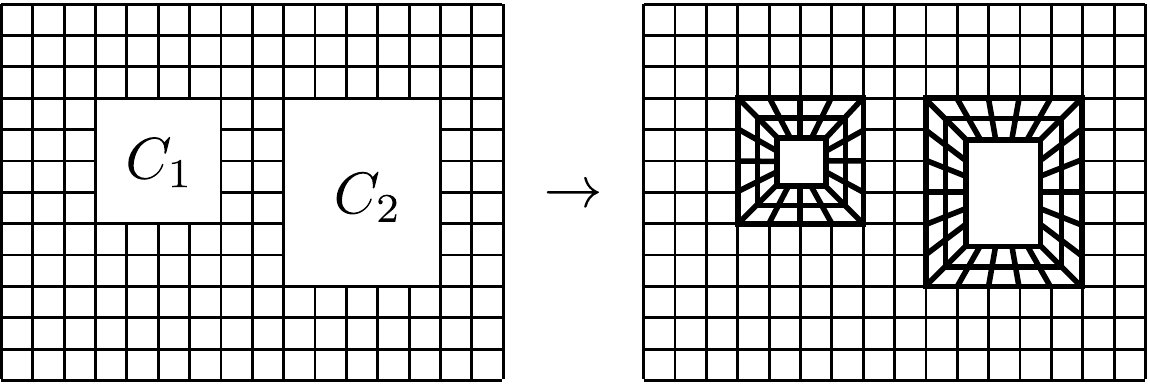}}
    \label{fig:framing_G}
  }
  \hspace{.5cm}
  \subfigure[$\{C_1,C_2\}$-Framing of a subgraph.]{
    \scalebox{0.72}{\includegraphics{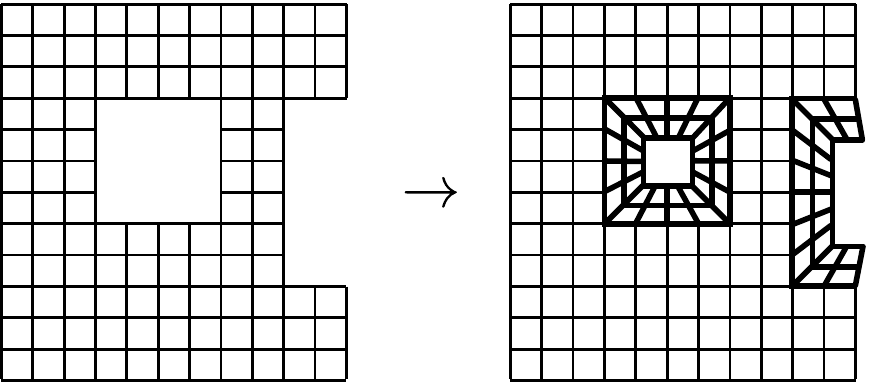}}
    \label{fig:framing_H}
  }
  \caption{Examples of graph framing.}
\end{center}
\end{figure}

We now state a basic property of framing whose proof follows directly
from the definition of a ${\cal C}$-framing.

\begin{lemma}\label{lem:framing_subgraphs}
  Let $G$ be a graph, and let ${\cal C}$ be a collection of cycles in
  $G$.  Let ${\cal H}$ be a collection of pairwise vertex-disjoint
  subgraphs of $G$.  Let $G'$ be the ${\cal C}$-framing of $G$, and
  for any $H\in {\cal H}$, let $H'$ be the ${\cal C}$-framing of $H$.
  Then, the graph $\bigcup_{H\in {\cal H}} H'$ is (isomorphic to) a
  subgraph of $G'$.
\end{lemma}

We first argue that framing does not increase the genus of the skeleton.

\begin{lemma}[The genus of a framed skeleton]\label{lem:framing_genus}
  Let $G$ be a graph, and let $\phi$ be a drawing of $G$ into some
  surface ${\cal S}$.  Let ${\cal P}$ be a collection of pairwise
  non-overlapping $\phi$-patches of $G$.  Let $G' = G \setminus \left(
    \bigcup_{(X,C)\in {\cal P}} X\setminus C \right)$.  Let $G''$ be
  the $\{C\}_{(X,C)\in {\cal P}}$-framing of $G'$.  Then, $\eg(G'')
  \leq \eg(G)$.
\end{lemma}
\iffull
\begin{proof}
  For every $(X,C)\in {\cal P}$, there exists a disk ${\cal D}_{(X,C)}
  \subset {\cal S}$, with $\partial {\cal D}_{(X,C)} = \phi(C)$, and
  such that $\phi(G)\cap {\cal D}_{(X,C)} = \phi(X)$.  Moreover, since
  the patches are pairwise non-overlapping, it follows that the
  resulting disks have pairwise disjoint interiors.  For every patch
  $(X,C)\in {\cal P}$, there exists a cylinder $K_{(X,C)}\subseteq
  G''$, such that the top of $K_{(X,C)}$ has been identified with $C$.
  We remove $X\setminus C$ from the drawing $\phi$, and we draw the
  rest of the cylinder $K_{(X,C)}$ inside ${\cal D}_{(X,C)}$.
  Repeating for all patches in ${\cal P}$, we obtain a drawing of
  $G''$ into ${\cal S}$.  This shows that $\eg(G'')\leq \eg(G)$, which
  concludes the proof.
\end{proof}
\fi

\begin{lemma}[Planarizing the skeleton]\label{lem:framed_planarization}
  There exists a polynomial-time algorithm which given a graph $G$ of
  maximum degree $\dmax$, an integer $g>0$, and a collection of
  pairwise non-overlapping patches $(X_1,C_1), \ldots, (X_r,C_r)$ of
  $G$ satisfying the assertion of Lemma \ref{lem:skeleton}, it either
  correctly decides that $\eg(G)>g$, or it outputs a set $S\subseteq
  V(G)$, satisfying the following conditions:
\begin{description}
\item{(1)}
$|S| = O(\dmax g^{12} \log^{19/2} n)$.
\item{(2)} Let ${\cal C} = \{C_1,\ldots,C_r\}$, and let $G' = G
  \setminus (\bigcup_{i=1}^r (X_i \setminus C_i))$.  For every
  connected component $H$ of $G'\setminus S$, we have that the ${\cal
    C}$-framing of $H$ is planar.
\end{description}
\end{lemma}
\iffull
\begin{proof}
  By Lemma \ref{lem:skeleton}, $G'$ has treewidth $t'=O(\dmax g^{11} \log^8 n)$.
  For a subgraph $H\subseteq G'$, let $H^{\framed}$ be the ${\cal
    C}$-framing of $H$.  We proceed to compute an auxiliary tree
  ${\cal T}$, where every vertex of ${\cal T}$ is a subgraph of $G'$,
  and we consider ${\cal T}$ as being rooted at $G'$.  We construct
  ${\cal T}$ inductively as follows: For every vertex $H\subseteq G'$
  of ${\cal T}$, if $H^{\framed}$ is planar, then we set $H$ to be a
  leaf of ${\cal T}$.  Otherwise, using the algorithm from Theorem
  \ref{thm:approx_treewidth} we find in polynomial time a
  $\Theta(1)$-balanced vertex separator $S_H \subseteq V(H)$, with
\[
|S_H| = O(t' \sqrt{\log n}) = O(\dmax g^{11} \log^{17/2} n).
\]
For every connected component $H'$ of $H \setminus S_H$, we add $H'$ to
$V({\cal T})$, and we make $H'$ a child of $H$.  This concludes the
definition of ${\cal T}$.  Clearly, ${\cal T}$ can be computed in
polynomial time.

The height of ${\cal T}$ is $h=O(\log n)$ since we use $\Theta(1)$-balanced
vertex separators to split each graph at an internal node of ${\cal T}$.
For any $i\in
\{0,\ldots,h\}$, let ${\cal L}_i$ be the set of vertices of ${\cal T}$
that are at level $i$ (i.e., ${\cal L}_0 = \{G'\}$).  Note that all
elements in any ${\cal L}_i$ are pairwise vertex-disjoint subgraphs of
$G$.  It follows by Lemma \ref{lem:framing_subgraphs} that for any
$i\in \{0,\ldots,h\}$,
\[
\bigcup_{H \in {\cal L}_i} H^{\framed} \subseteq (G')^{\framed}.
\]
Note that all non-leaf subgraphs of $G'$ in ${\cal L}_i$ are
non-planar.  Therefore, by Lemma \ref{lem:framing_genus} there can be
at most $O(g)$ non-leaf subgraphs in ${\cal L}_i$.  If this is not the
case, then we can correctly decide that $\eg(G)>g$.  Otherwise, let
$S$ be the set of all vertex separators computed throughout the
construction of ${\cal T}$, i.e.~$S=\bigcup_{H}S_H$, where $H$ ranges
over all non-leaf vertices of ${\cal T}$.  We have
\begin{align*}
|S| &= \sum_H |S_H|\\
 &\leq \sum_{i=0}^h |{\cal L}_i| \cdot O(\dmax g^{11} \log^{17/2} n)\\
 &= O(\dmax g^{12} \log^{19/2} n),
\end{align*}
as required.
\end{proof}
\fi

We can now prove the main result of this paper.

\begin{proof}[\proofof Theorem \ref{thm:main}]
  By Lemma \ref{lem:skeleton}, in polynomial time, we can either
  correctly decide that $\eg(G)>g$, or we can compute a (possibly
  empty) collection ${\cal P}$ of pairwise non-overlapping patches of
  $G$, satisfying the following conditions:
\begin{itemize}
\item If $\eg(G)\leq g$, then there exists a drawing $\phi$ of $G$ into a surface ${\cal S}$
  of Euler genus $g$, such that any $(X,C)\in {\cal P}$ is a
  $\phi$-patch.
\item
The graph
$G' = G \setminus \left( \bigcup_{(X,C)\in {\cal P}} X\setminus C \right)$
has treewidth $t'=O(\dmax g^{11} \log^{8} n)$.
\end{itemize}
Let
${\cal C} = \{C : (X,C)\in {\cal P}\}$.
For any subgraph $H\subseteq G'$, let $H^{\framed}$ denote the ${\cal C}$-framing of $H$.

Let $S\subseteq V(G)$ be the set computed by Lemma
\ref{lem:framed_planarization}.  Let $G'' = G' \setminus S$, and let
${\cal H}$ be the set of connected components of $G''$.  Since for
every $H \in {\cal H}$ the graph $H^{\framed}$ is planar, it follows
that $(G'')^{\framed} = \bigcup_{H\in {\cal H}} H^{\framed}$ is also
planar.  Pick a planar drawing $\psi$ of $(G'')^{\framed}$ (which can
be computed, e.g.~by the algorithm of Hopcroft, and Tarjan
\cite{Hopcroft}).

We now proceed to obtain a drawing of $G$, by modifying the drawing
$\psi$ of $(G'')^{\framed}$.  We iterate over all patches $(X,C)\in
{\cal P}$.  Consider some $(X,C)\in {\cal P}$.  Since every $(X,C)\in
{\cal P}$ is a $\phi$-patch, it follows that the graph $X$ admits a
planar drawing $\phi_{(X,C)}$ into a disk ${\cal D}_{(X,C)}$, with
$\phi_{(X,C)}(C) = \partial {\cal D}_{(X,C)}$.  Let
\[
n_C = |E(C) \setminus E(G'')|.
\]
We consider two cases:
\begin{description}
\item{(i)} If $n_C = 0$, it follows that $C\subseteq G''$.  Therefore,
  $(G'')^{\framed}$ contains a $(3\times |V(C)|)$-cylinder $K$.  The
  top of $K$ is identified with $C$ in $(G'')^{\framed}$.  Since $K$
  is 3-vertex-connected, it admits a unique planar drawing.  Let $C'$
  be the bottom of $K$.  It follows $K$ bounds a face $F$ in $\psi$.
  We can therefore extend the current drawing to $X\setminus C$, by
  placing the open disk ${\cal D}_{(X,C)} \setminus \partial {\cal
    D}_{(X,C)}$ inside the face $F$, deleting all vertices of $K$ that
  do not belong to its top, and connecting the vertices in $X$, with
  their neighbors in the copy of $C$ in
  $G''$. 
  Notice that in this case, we do not increase the genus of the
  current surface.
\begin{center}
\scalebox{0.72}{\includegraphics{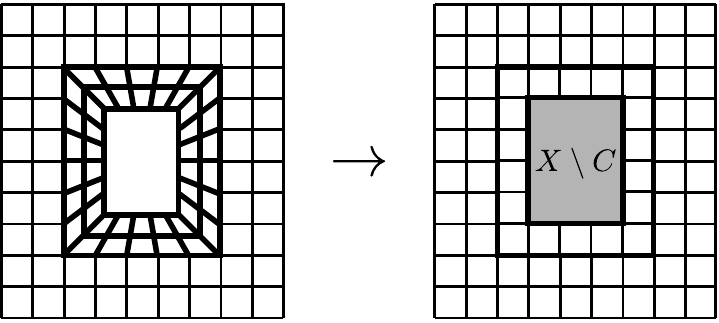}}
\end{center}

\item{(ii)} If $n_C > 0$, we proceed as follows.  First, we place the
  disk ${\cal D}_{(X,C)}$ in the unbounded face of the current
  drawing.  Let $P_1,\ldots,P_{n_C}$ be the set of maximal segments of
  $C$ contained in $G''$.  Consider a maximal segment $P_j$.  There
  exists a $(3\times |V(P_j)|)$-grid $L$ in $(G'')^{\framed}$, such
  that the top of $L$ is identified with $P_j$.  We argue that the
  bottom of $L$ is a segment of a face $F$ in $\psi$:
  If $|V(P_j)| \leq 2$, this is immediate, and if $|V(P_j)| \geq 3$,
  it follows by Whitney's theorem, since $L$ is 3-connected, and therefore has a unique planar drawing.
    If the
  orientation of $P_j$ along a clockwise traversal of $\partial {\cal
    D}_{(X,C)}$ agrees with the orientation of $P_j$ along a clockwise
  traversal of $F$, then we attach a handle between $P_j$ in $X$, and
  the bottom of $L$.  Otherwise, we attach a M\"{o}bius band.  Next,
  we delete all the vertices in $L$ that do not belong to its bottom,
  and we also delete the copy of $P_j$ in $X$ (i.e.~the copy that lies
  on the boundary of ${\cal D}_{(X,C)}$.  Finally, we draw the edges
  between $X\setminus C$, and $P_j$, by routing them along the new
  handle, or M\"{o}bius band.  We continue in the same fashion, with
  all the remaining maximal segments. 
  For every maximal segment, we increase the Euler genus of the
  underlying surface by at most $3$.  Therefore, the total increase
  in the Euler genus is at most $3 n_{C}$.
\begin{center}
\iffull
\scalebox{0.72}{\includegraphics{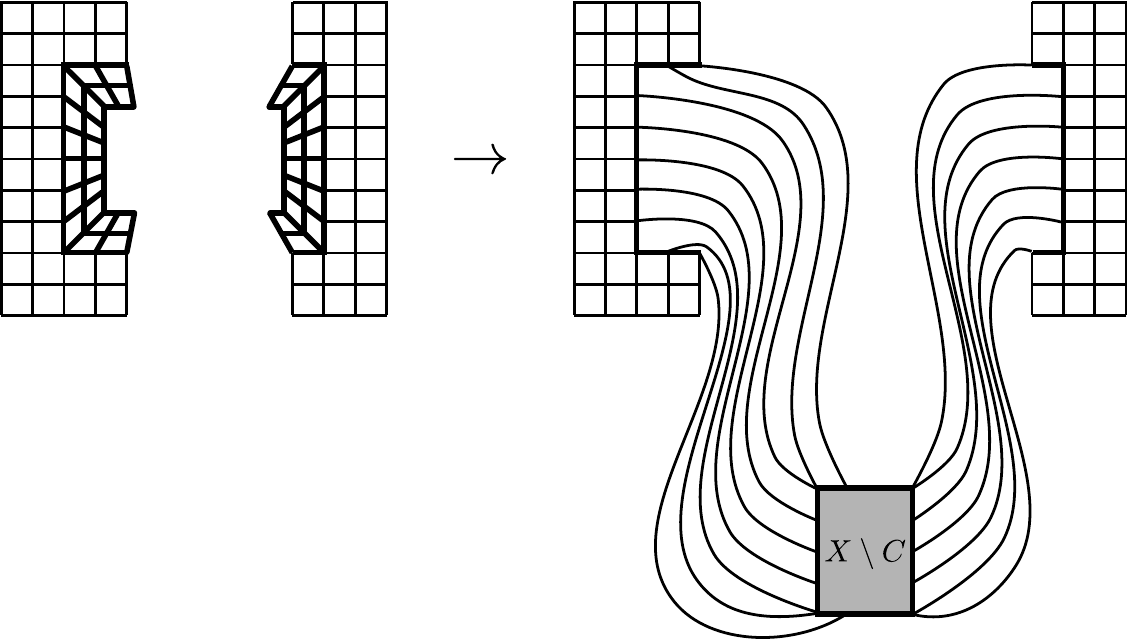}}
\fi
\ifabstract
\scalebox{0.6}{\includegraphics{figs/extension_2}}
\fi
\end{center}
\end{description}
After considering all patches in ${\cal P}$, we arrive at a drawing
into some surface.  We remove any remaining vertices from
$(G'')^{\framed} \setminus G''$.  We arrive at a drawing $\psi'$ of
the graph
\[
\Gamma = G\left[V(G'') \cup \left(\bigcup_{(X,C)\in {\cal P}} X\setminus C \right) \right]  = G \setminus S.
\]

Since the cycles $\{\phi(C)\}_{C\in {\cal C}}$ bound disks with
disjoint interiors in the drawing $\phi$, it follows that every edge of
$G'$ is contained in at most two cycles in ${\cal C}$.  Therefore,
\begin{align*}
  \sum_{C\in {\cal C}} n_C &\leq 2 \cdot |E(G')\setminus E(G'')| \leq \dmax \cdot |V(G') \setminus V(G'')|\\
  &= \dmax \cdot |S| = O(\dmax^2 g^{12} \log^{19/2}n).
\end{align*}
It follows that the resulting drawing $\psi'$ of $\Gamma$ is into a
surface of Euler genus at most $\sum_{C\in {\cal C}} 3 n_C =
O(\dmax^2 g^{12} \log^{19/2}n)$.

It remains to extend the drawing to $S$.  This can be done by adding
at most $|S|\cdot \dmax$ additional handles (one for every edge
incident to a vertex in $S$).  The resulting drawing of $G$ has Euler
genus $O(\dmax^2 g^{12} \log^{19/2}n) + O(|S|\cdot \dmax) = O(\dmax^2 g^{12} \log^{19/2}n)$, as required.  This concludes the proof.
\end{proof}

\section{Computing a skeleton}\label{sec:skeleton}

The rest of the paper is devoted to the algorithm for computing a
skeleton, i.e.~the proof of Lemma \ref{lem:skeleton}.  At the high
level, the algorithm proceeds iteratively as follows: (i) We compute a
patch. (ii) We remove its interior. (iii) We repeat until the
treewidth of the remaining graph becomes small enough.  One issue with
implementing the above approach is that we do not have access to an
explicit optimal drawing of the input graph.  Therefore, we cannot
argue that a computed patch is a $\phi$-patch for a \emph{specific}
optimal drawing.  To that end, we will need a stronger notion of a
patch.  More precisely, we need to compute patches that are
$\phi$-patches, in \emph{any} optimal drawing $\phi$.  Moreover,
because the above procedure computes patches in an ever decreasing
graph, we need to make sure that the genus never decreases.  This
property will ensure that at the end of the procedure, all computed
patches are $\phi$-patches for \emph{the same} optimal drawing $\phi$.
The following definition states precisely the properties that we need.

\begin{definition}[Universal patch]\label{defn:universal_patch}
  Let $G$ be a graph of Euler genus $g$.  Let $X \subseteq G$ be a
  subgraph of $G$, and let $C \subsetneq X$ be a cycle in $X$.  We say
  that $(X,C)$ is a \emph{universal patch} (of $G$) if it satisfies
  the following conditions:
\begin{description}
\item{(1)} For {\em any} drawing $\psi$ of $G$ into a surface of Euler
  genus $g$, we have that $(X,C)$ is a $\psi$-patch.
\item{(2)}
Let $G'=G \setminus (X \setminus C)$.
Then, $\eg(G')=\eg(G)=g$.
\end{description}
\end{definition}

The following lemma shows that we can compute a universal patch in
polynomial time, in a normalized graph of sufficiently large
treewidth.  A crucial property of the algorithm is that after removing
the interior of the computed patch, the resulting graph remains
normalized.  This fact allows us to inductively maintain a normalized
graph, while computing the skeleton.  Intuitively, dealing with a
normalized graph is essential for avoiding ``locally-pathological''
planar drawings.  Roughly speaking, a non-normalized graph can have an
optimal drawing that locally looks rather complicated.  This makes it
very difficult to control how different overlapping patches interact
with each other.

\begin{lemma}[Computing a universal patch]\label{lem:computing_a_universal_patch}
  There exists a universal constant $\alpha>0$, such that the
  following holds. Let $G$ be a normalized graph of Euler genus $g\geq
  1$, treewidth $t\geq 1$, and maximum degree $\dmax$. Suppose that $t
  \geq \alpha \dmax g^{11} \log^{15/2} n$. Then, we can compute in
  polynomial time a universal patch $(X,C)$ in $G$, such that $G
  \setminus (X\setminus C)$ is normalized.
\end{lemma}

The proof of Lemma \ref{lem:computing_a_universal_patch} is rather
long, and requires several other results, including our algorithm for computing a flat grid minor, and properties of planarly nested sequences due to Mohar \cite{Mohar_local_planarity}.
For that reason, we defer
it to subsequent sections, and show first how to use it to construct a
skeleton (i.e.~to prove Lemma \ref{lem:skeleton}).

Before we proceed with the proof of Lemma \ref{lem:skeleton}, we first
derive a property that will be used in showing the correctness of the
algorithm.  While computing a sequence of patches, it is possible that
a patch overlaps a previously computed one.  In this case, we can show
that we can essentially ``merge'' the two patches.  The following
Lemma gives the necessary properties for doing exactly that.

\begin{lemma}[Merging overlapping patches]\label{lem:nested_patches}
  Let $G$ be a graph, let $\phi_1$ be a drawing of $G$ into a surface
  ${\cal S}$, and let $(X_1,C_1)$ be a $\phi_1$-patch of $G$.  Let $G'
  = G \setminus (X_1 \setminus C_1)$.  Let $\phi_2$ be the drawing of
  $G'$ into ${\cal S}$ obtained by restricting $\phi_1$ to $G'$.  Let
  $(X_2, C_2)$ be a $\phi_2$-patch of $G'$.  Suppose further that
  $(V(X_2) \setminus V(C_2)) \cap V(C_1) \neq \emptyset$.  Then, $(X_1
  \cup X_2, C_2)$ is a $\phi_1$-patch of $G$.
\end{lemma}

\begin{proof}
  Since $(X_2,C_2)$ is a $\phi_2$-patch, it follows that there exists
  ${\cal D}_2 \subset {\cal S}$, such that $\partial {\cal D}_2 =
  \phi_2(C_2)$.  It remains to show that $\phi_1(G) \cap {\cal D}_2 =
  \phi_1(X_1 \cup X_2)$.  Since $(X_1,C_1)$ is a $\phi_1$-patch, it
  follows that there exists a disk ${\cal D}_1 \subset {\cal S}$, with
  $\partial {\cal D}_1 = \phi_1(C_1)$.  We claim that ${\cal D}_1
  \subseteq {\cal D}_2$.  Suppose to the contrary that there exists a
  point $p \in {\cal D}_1 \setminus {\cal D}_2$.  Pick an arbitrary
  vertex $v \in (V(X_2) \setminus V(C_2)) \cap V(C_1)$, and let $q =
  \phi_1(v) = \phi_2(v)$.  We have $q \in {\cal D}_2 \cap {\cal D}_1$.
  Moreover, since $v\in V(X_2)\setminus V(C_2)$, it follows that $q
  \notin \partial {\cal D}_2$.  Therefore, there exists a segment of
  $\partial {\cal D}_2$ that lies inside ${\cal D}_1$.  This implies
  that $\phi_2(X_2)$ intersects the interior of ${\cal D}_1$.  Thus,
  there exists $e \notin E(X_1)$, with $\phi_2(e) = \phi_1(e) \subset
  {\cal D}_1$, which contradicts the fact that $(X_1,C_1)$ is a
  $\phi_1$-patch.
\end{proof}

We are now ready to prove Lemma \ref{lem:skeleton}, which is the main result of this section.
Before proceeding, we remark that the assertion of Lemma \ref{lem:skeleton} can in fact be slightly strengthened.
More precisely, one can show that in the computed collection, all patches are universal.
We omit the details since they are not relevant to our algorithmic application.

\begin{proof}[\proofof Lemma \ref{lem:skeleton}]
Fix a drawing $\phi$ of $G$ into a surface of Euler genus $g$.
We remark that we use $g$ in the following argument, even though we do not know how to explicitly compute it in polynomial time.

We inductively compute a sequence ${\cal P}^{(0)},\ldots,{\cal P}^{(s)}$, with ${\cal P}^{(0)}=\emptyset$, where for each $i\in \{1,\ldots,s\}$, ${\cal P}^{(i)}$ is a collection of pairwise non-overlapping $\phi$-patches of $G$.
The desired collection will be ${\cal P}^{(s)}$.

For any $\ell\in \{1,\ldots,s\}$, we define the graph
$G^{(\ell)} = G \setminus \bigcup_{(X,C)\in {\cal P}^{(\ell)}} (X \setminus C)$.
We maintain the inductive invariant that for any $\ell \in \{1,\ldots,s\}$, 
\begin{align}
G^{(\ell)} \mbox{ is normalized, and } \eg(G^{(\ell)}) = \eg(G) = g. \label{eq:skeleton_invariant}
\end{align}

Given ${\cal P}^{(\ell)}$, for some $\ell\geq 0$, we proceed as
follows.  Let $\alpha>0$ be the constant in the statement of Lemma
\ref{lem:computing_a_universal_patch}.  By Theorem
\ref{thm:approx_treewidth}, there exists a polynomial time algorithm
that given a graph of treewidth $k$, outputs a tree decomposition of
width at most $\alpha' k \sqrt{\log k}$, for some universal constant
$\alpha'$.  We run the algorithm of Theorem \ref{thm:approx_treewidth} on
$G^{(\ell)}$.  If the algorithm returns a tree decomposition of width
at most $\alpha \cdot \alpha' \cdot \dmax g^{11} \log^{8} n$, then
we have a certificate that the treewidth of $G^{(\ell)}$ is at most
$O(\dmax g^{11} \log^{8} n)$, and we set $r=\ell$.  Otherwise, we
know that the treewidth of $G^{(\ell)}$ is at least $\alpha \dmax g^{11} \log^{15/2} n$, and we proceed to compute ${\cal P}^{(\ell+1)}$.

By Lemma \ref{lem:computing_a_universal_patch}, and since $G^{(l)}$ is
normalized, we can compute in polynomial time a universal patch in
$(X^*,C^*)$ of $G^{(\ell)}$.  Let
\[
{\cal Q}^{(\ell)} = \{(X,C) \in {\cal P}^{(\ell)} : (X,C) \mbox{ and } (X^*,C^*) \mbox{ are overlapping}\}.
\]
Fix an ordering of the patches in ${\cal Q}^{(\ell)} = \{(Y_i,F_i)\}_{i=1}^{k_{\ell}}$, where $k_{\ell}=|{\cal Q}^{(\ell)}|$.
Let 
$Y = \bigcup_{i=1}^{k_{\ell}} Y_i$.
We argue that $(X^*\cup Y, C)$ is a $\phi$-patch of $G$.
Let $\Gamma^{(0)} = G^{(\ell)}$, and for any $j\in \{1,\ldots,k_{\ell}\}$, let 
$\Gamma^{(j)} = \Gamma^{(j-1)} \cup Y_j$.
Let also $\phi^{(j)}$ be the drawing of $\Gamma^{(j)}$ induced by
restricting $\phi$ to $\Gamma^{(j)}$.  Since $(X^*, C^*)$ is a
universal patch of $\Gamma^{(0)} = G^{(\ell)}$, and $\phi^{(0)}$ is a
drawing into a surface of Euler genus $g$, it follows that $(X^*,C^*)$
is also a $\phi^{(0)}$-patch.  Note that for any $j\in
\{1,\ldots,k_{\ell}\}$, since $(Y_j,F_j)$ is a $\phi$-patch, and
$\phi^{(j)}$ is a restriction of $\phi$, it follows that $(Y_j, F_j)$
is also a $\phi^{(j)}$-patch.  By Lemma \ref{lem:nested_patches} we
obtain that $(X^* \cup Y_1, C^*)$ is a $\phi^{(1)}$-patch of
$\Gamma^{(1)}$.  By inductively applying Lemma
\ref{lem:nested_patches} on the pair of patches $(Y_j,F_j)$, and
$\left(X^* \cup \left( \bigcup_{i=1}^{j-1} Y_i \right), C^*\right)$,
we conclude that $(X^*\cup Y, C)$ is a $\phi^{(k_{\ell})}$-patch of
$\Gamma^{(k_{\ell})}$.  Since the patch $(X^*\cup Y, C)$ is
non-overlapping with any of the patches in ${\cal P}^{(\ell)}
\setminus {\cal Q}^{(\ell)}$, it follows that $(X^*\cup Y, C)$ is a
$\phi$-patch.  We set
\[
{\cal P}^{(\ell+1)} = \left({\cal P}^{(\ell)} \setminus {\cal Q}^{(\ell)}\right) \cup \{ (X^* \cup Y, C^*)\}.
\]
It is immediate that ${\cal P}^{(\ell+1)}$ is a collection of pairwise
non-overlapping $\phi$-patches.

We next show that the inductive invariant
\eqref{eq:skeleton_invariant} is maintained.  Observe that
\begin{align*}
G^{(\ell+1)} &= G^{(\ell)} \setminus (X^* \setminus C^*).
\end{align*}
Since $(X^*,C^*)$ is a universal patch of $G^{(l)}$, it follows that
$\eg(G^{(l+1)})=\eg(G^{(l)}) = g$.  Moreover, by Lemma
\ref{lem:computing_a_universal_patch} we have that $G^{(\ell+1)}$ is
normalized.  This shows that the inductive invariant
\eqref{eq:skeleton_invariant} is maintained.

It remains to argue that the above process terminates after
polynomially many steps.  Note that since $(X^*,C^*)$ is a patch, we
have $X^* \subsetneq C^*$.  Therefore, $G^{(\ell+1)} \subsetneq
G^{(\ell)}$.  It follows that the algorithm terminates in polynomial
time.  This concludes the proof.
\end{proof}

\section{Computing a flat grid minor}\label{sec:flat_grid_minors}

The last ingredient required for our algorithm is a procedure for
computing a universal patch (the proof of Lemma
\ref{lem:computing_a_universal_patch}).  This procedure requires a
polynomial-time algorithm which given a graph $G$ of large treewidth,
and small Euler genus, computes a large \emph{flat} grid minor in $G$.
Intuitively, we say that a subgraph $H$ is flat if it is planar, and
moreover it admits a planar drawing, such that all edges leaving $H$,
are incident to its outer face.  A formal definition can be found
later in this Section.  We give such an algorithm for computing flat
grid minors in the present Section.  The proof of Lemma
\ref{lem:computing_a_universal_patch} appears in the following
Section.

We remark that all previously known algorithm for computing flat grid
minors either work when $G$ is of constant treewidth, or require a
drawing of $G$ into a small-genus surface as part of the input.  In
contrast, we need to allow the treewidth to be as large as
$\Omega(n)$, and we of course do not have access to a low-genus
drawing of $G$ -- after all, computing such a drawing is precisely our
goal!

Let us now give a high-level overview of our algorithm for finding a
flat grid minor.  It is known that any graph $G$ of treewidth $t$, and
genus $g$, contains a $(\Omega(t/g) \times \Omega(t/g))$-grid minor.
By repeatedly removing balanced vertex-separators, we can compute in
polynomial time a set $X$ of at most $O(tg\log^{O(1)}n)$ vertices,
whose removal leaves a planar graph.  In particular, we can show 
that the resulting planar graph $G\setminus X$ must still contain a
relatively large grid minor.  It is already known how to compute a
large grid minor in a planar graph in polynomial time.  Unfortunately,
such a grid minor is not guaranteed to be flat in the original graph
$G$.  We argue that a large \emph{subgrid} of the computed grid minor
must be flat $G$.

\ifabstract
We now formally define the notion of \emph{flatness}.

\begin{definition}[Flatness]
  Let $G$ be a graph, and let $H \subseteq G$ be a planar subgraph.
  We say that $H$ is \emph{flat} (w.r.to $G$) if there exists a planar
  drawing of $H$, such that for all edges $\{u,v\} \in E(G)$, with
  $u\in V(H)$, and $v\in V(G) \setminus V(H)$, the vertex $u$ is in
  the outer face of $H$.
\end{definition}

We state the main lemma of this section and refer the reader 
to the full version for the details of the proof.
\fi

\iffull
We now proceed with the formal proof.  The first step towards
computing a flat grid minor, is a procedure for computing a small vertex
set, whose removal leaves a planar graph.  This is described in the
following lemma.


\begin{lemma}[Computing a small planarizing set]\label{lem:planarization}
  There exists a polynomial-time algorithm which given a graph $G$ of
  treewidth $t$, and an integer $g\geq 0$,
  either correctly decides that $\eg(G)>g$, or it outputs a set
  $X\subseteq V(G)$, satisfying the following conditions:
\begin{description}
\item{(1)}
$|X| = O(g t \log^{5/2} n)$.
\item{(2)}
$G\setminus X$ is planar.
\end{description}
\end{lemma}
\iffull
\begin{proof}
  We inductively define a collection ${\cal F}$ of subgraphs of $G$,
  and an auxiliary tree $T$ with $V(T)={\cal F}$.  Initially, we set
  ${\cal F}_0=\{G\}$, and we let $T_0$ be the tree with single vertex
  $G$, which we also consider to be its root.  Given ${\cal F}_i$, and
  $T_i$, we proceed as follows.  If all leaves of $T_i$ are planar
  subgraphs of $G$, then we set ${\cal F}={\cal F}_i$, and $T = T_i$,
  concluding the construction of ${\cal F}$, and $T$.  Otherwise, let
  ${\cal L}_i$ be the set of leaves of $T$ that are non-planar
  subgraphs of $G$.  Each $H \in {\cal L}_i$ is a subgraph of $G$, and
  therefore has treewidth at most $t$.  Therefore, $H$ contains a
  $\Theta(1)$-balanced vertex separator $S_H$, with $|S_H|\leq t$.
  Using the algorithm from Theorem \ref{thm:approx_treewidth}, we can
  compute a $\Theta(1)$-balanced vertex separator $S'_H$ of $H$, with
  $|S'_H| = O(|S_H| \sqrt{\log n}) = O(t \sqrt{\log n})$.  Let ${\cal
    C}_H$ be the set of connected components of $H\setminus S'_{H}$.
  We set
\[
{\cal F}_{i+1} = {\cal F}_i \cup \bigcup_{H\in {\cal L}_i} {\cal C}_H
\]
We also construct $T_{i+1}$ as follows.
Initially, we set $T_{i+1}=T_i$.
For every $C\in {\cal C}_H$, we add the vertex $C$ in $V(T_{i+1})$, and we add the edge $\{H,C\}\in E(T)$.
Note that $H$ is the parent of $C$ in the rooted tree $T_{i+1}$.
We continue until all the leaves of the current tree are planar subgraphs of $G$.
This completes the definition of ${\cal F}$, and $T$.

Let ${\cal L}$ be the set of leaves of $T$.
We can now define
\[
X = \bigcup_{H \in {\cal F} \setminus {\cal L}} S'(H).
\]
It is immediate by the construction that $G\setminus X = \bigcup_{H\in {\cal L}} H$.
Since all leaves of $T$ are planar subgraphs of $G$, it follows that every connected component of $G\setminus X$ is planar, and therefore $G\setminus X$ is planar.

It remains to bound $|X|$.
Observe that for any $i>0$, every graph in ${\cal F}_i$ is a connected component of $H\setminus S'_H$, for some $H\in {\cal F}_{i-1}$.
Since $S'_H$ is a $\Theta(1)$-balanced separator of $H$, it follows that the maximum size of every graph in ${\cal F}_i$ is at most $2^{-i c }$, for some constant $c>0$.
Therefore, the depth of $T$ is $h=O(\log n)$.

For every $i\in \{0,\ldots,h\}$, let ${\cal Z}_i$ be the set of all graphs $H\in {\cal F}$, such that $H$ is at depth $i$ in $T$, and it is non-planar.
For any $i\in \{0,\ldots,h\}$, we have
\begin{align*}
\bigcup_{H\in {\cal Z}_i} H &\subseteq G.
\end{align*}
Since all graphs in ${\cal Z}_i$ are pairwise vertex-disjoint subgraphs of $G$, we have
\begin{align*}
\sum_{H\in {\cal Z}_i} \eg(H) &\leq \eg(G) = g.
\end{align*}
Since all graphs in ${\cal Z}_i$ are non-planar, we have that for any $H\in {\cal Z}_i$, $\eg(H)\geq 1$.
Therefore, 
\begin{align*}
 |{\cal Z}_i| &\leq g.
\end{align*}
We have
\begin{align*}
 |X| &= \left| \bigcup_{H \in {\cal F} \setminus {\cal L}} S'(H) \right| \\
  &= \left| \bigcup_{i=0}^h \bigcup_{H\in {\cal Z}_i} S'(H) \right| \\
  &\leq O(h g t \log^{3/2} n)\\
  &\leq O(g t \log^{5/2} n),
\end{align*}
as required.
\end{proof}
\fi

The following lemma gives a trade-off between the treewidth of a genus
$g$ graph, and the size of its largest grid minor.



\begin{lemma}[Demaine et al.~\cite{DBLP:journals/jacm/DemaineFHT05}]\label{lem:grid_genus}
  Let $G$ be a graph of Euler genus $g$, and treewidth $t$.  Then, $G$
  contains a $(r\times r)$-grid minor, for some $r\geq
  \frac{t+1}{6(g+1)}$.
\end{lemma}


The following lemma shows that when deleting a small subset of a
graph, a moderately large grid minor ``survives''.  This property will
be used crucially by our algorithm for computing a grid minor.  An
earlier version of the present paper gave a slightly weaker bound on
the size of the resulting grid minor.  Our bound has been subsequently
improved by Eppstein \cite{Eppstein_grids}, who obtained an
asymptotically optimal estimate.

\begin{lemma}[Eppstein \cite{Eppstein_grids}]\label{lem:grids_persistency}
  Let $r,f \geq 1$.  Let $G$ be the $(r \times r)$-grid, and $X\subset
  V(G)$, with $|X| = f$.  Then, $G\setminus X$ contains the $(r'
  \times r')$-grid as a minor, where $r' = \max\{ r-f, r^2/4f - O(1)
  \} = \Theta(\min\{r, r^2/f\})$.
\end{lemma}

The following lemma is an intermediate step towards getting a flat
grid minor.  The main missing property is that the grid minor
guaranteed by Lemma \ref{lem:large_tw_planar_piece} might not be flat.
We will subsequently ensure flatness via a more careful argument.

\begin{lemma}\label{lem:large_tw_planar_piece}
  Let $G$ be a graph of Euler genus $g\geq 1$, and treewidth $t\geq
  1$. There is a polynomial time algorithm to compute a set
  $X\subseteq V(G)$, with $|X| = O(gt \log^{5/2} n)$, and a planar
  connected component $\Gamma$ of $G \setminus X$ containing the
  $(r'\times r')$-grid as a minor, with $r' = \Omega\left(
    \frac{t}{g^{3} \log^{5/2}n} \right)$.  (The algorithm does not
  require a drawing of $G$ as part of the input.)
\end{lemma}
\iffull
\begin{proof}
  By lemma \ref{lem:grid_genus}, we have that $G$ contains as a minor
  the $(r \times r)$-grid $H$, for some $r\geq \frac{t+1}{6(g+1)} >
  \frac{t}{16g}$.  More precisely, for every $v \in V(H)$, there
  exists $U_v \subseteq V(G)$, such that the following conditions are
  satisfied:
\begin{itemize}
 \item
 For any $v, v' \in V(H)$, $U_{v} \cap U_{v'} = \emptyset$.

 \item
 For any $\{v, v'\} \in E(H)$, there exist $u \in U_{v}$, and $u' \in U_{v'}$, such that $\{u,u'\} \in E(G)$.
 \end{itemize}

Using the algorithm from Lemma \ref{lem:planarization}, we compute a set $X\subseteq V(G)$, with $|X| = O(gt \log^{5/2} n)$, such that the graph $G \setminus X$ is planar.
We first argue that $G \setminus X$ contains a large grid minor.
Let
\begin{align*}
X' = \{v \in V(H) : X \cap U_v \neq \emptyset\}.
\end{align*}
By Lemma \ref{lem:grids_persistency}, we have that $H \setminus X'$ contains as a subgraph a $(r' \times r')$-grid $H'$, for some
\begin{align*}
r' &= \Omega(\min\{r, r^2/f \})\\
 &= \Omega\left(\min\left\{\frac{t}{g}, \frac{t^2}{g^3 t \log^{5/2} n} \right\} \right)\\
 &= \Omega\left(\frac{t}{g^3 \log^{5/2} n}\right).
\end{align*}
Since $H \setminus X'$ is a minor of $G \setminus X$, and $H'$ is a subgraph of $H$, it follows that $H'$ is a minor of $G \setminus X$.
By setting $\Gamma$ to be the connected component of $G \setminus X$ containing $H'$ as a minor, the assertion follows.
\end{proof}
\fi


The following result gives a polynomial-time constant-factor
approximation for computing a maximum grid minor in a planar graph.
It is obtained by combining the algorithm for computing a branch
decomposition due to Seymour, and Thomas
\cite{DBLP:journals/combinatorica/SeymourT94}, with the grid minor
algorithm of Robertson, Seymour, and Thomas
\cite{DBLP:journals/jct/RobertsonST94} (see also
\cite{DBLP:conf/focs/ChekuriKS04,GuT11}).

\begin{lemma}[Seymour, and Thomas \cite{DBLP:journals/combinatorica/SeymourT94}, Robertson, Seymour, and Thomas \cite{DBLP:journals/jct/RobertsonST94}]\label{lem:planar_grid_minor_approx}
  Let $r>0$, and let $G$ be a planar graph containing a $(r \times
  r)$-grid minor.  Then, on input $G$, we can compute in polynomial
  time a $(\Omega(r)\times \Omega(r))$-grid minor in $G$.
\end{lemma}


We now formally define the notion of \emph{flatness}.

\begin{definition}[Flatness]
  Let $G$ be a graph, and let $H \subseteq G$ be a planar subgraph.
  We say that $H$ is \emph{flat} (w.r.to $G$) if there exists a planar
  drawing of $H$, such that for all edges $\{u,v\} \in E(G)$, with
  $u\in V(H)$, and $v\in V(G) \setminus V(H)$, the vertex $u$ is on 
  the outer face of the drawing of $H$.
\end{definition}

Our strategy for computing a flat grid minor is as follows.  We first
remove a small set of vertices from the input graph $G$, so that the
resulting graph is planar, and has large treewidth.  This means that
it also has a large grid minor.  We argue that since we only remove a
small number of vertices from $G$, some sub-grid of this grid-minor
must be flat.

\begin{lemma}[Computing a flat grid minor]\label{lem:flat_grid}
There exists a polynomial-time algorithm which
  given a graph $G$ of treewidth $t$, and maximum degree $\dmax$, and
  an integer $g\geq 1$,
  either correctly decides that $\eg(G)>g$, or it outputs a flat
  subgraph $G'\subset G$, such that $X$ contains a $\left(\Omega(r)
    \times \Omega(r)\right)$-grid minor $M$, for some
  $r=\Omega\left(\frac{t^{1/2}}{\dmax^{1/2} g^{7/2} \log^{15/4}
      n}\right)$.  Moreover, in the latter case, the algorithm also
  outputs a minor mapping for $M$.
\end{lemma}

\iffull
\begin{proof}
  We first use Lemma \ref{lem:large_tw_planar_piece} to find a set $X \subseteq V(G)$, with $|X| = O(gt \log^{5/2} n)$, and  planar connected component $\Gamma$ of $G\setminus X$, such that $\Gamma$ contains a
  $(r'\times r')$-grid minor, for some
  $r'=\Omega\left(\frac{t}{g^{3} \log^{5/2} n}\right)$.  Using Lemma \ref{lem:planar_grid_minor_approx} we can
  compute a $(k\times k)$-grid minor $H$ in $\Gamma$, for some
  $k=\Omega(r')$.

Fix a minor mapping $\mu:V(H) \to 2^{V(\Gamma)}$ for $H$. 
Note that we can choose $\mu$ so that $\mu(H)=\Gamma$.
The grid $H$ contains $\ell=|X|\cdot \dmax + 1$ pairwise vertex-disjoint $(k'\times k')$-grids $H_1,\ldots,H_\ell$, for some $k'=\Omega\left(\frac{k}{\ell^{1/2}}\right) = \Omega\left(\frac{t^{1/2}}{\dmax^{1/2} g^{7/2} \log^{15/4} n}\right)$.
Since $G$ has maximum degree $\dmax$, the set $X$ is adjacent to at most $\dmax \cdot |X|$ vertices in $G \setminus X$.
It follows that there exists $i\in \{1,\ldots,\ell\}$, such that $\mu(H_i)$ is not adjacent to $X$.
It follows that the neighborhood of $\mu(H_i)$ is contained in $H$, which implies that $\mu(H_i)$ is flat, concluding the proof.
\end{proof}

\fi

\section{Computing a universal patch}\label{sec:universal_patch}

In this Section, we present the last missing ingredient of our drawing
algorithm: a procedure for computing a universal patch.  Our algorithm
for computing universal patches uses as a subroutine the procedure for
computing a flat grid minor from Section \ref{sec:flat_grid_minors}.

Our proof uses powerful machinery developed by Mohar
\cite{Mohar_local_planarity}.  The tools from
\cite{Mohar_local_planarity} allow us to guarantee that certain cycles
of a given graph are contractible in \emph{any} optimal drawing.  This
property will be crucial when computing a universal patch.  We begin
with some definitions.

Let $H$ be a subgraph of a graph $G$.  An \emph{$H$-component} of $G$
is either an edge in $E(G)$ with both endpoints in $V(H)$, or it is a
connected component $X$ of $G \setminus V(H)$ together with all the
edges between $X$ and $H$.  Each edge of an $H$-component $Y$ with an
endpoint in $V(H)$ is a \emph{foot} of $Y$.

\begin{definition}[Planarly nested sequence \cite{Mohar_local_planarity}]
  A sequence $C_1,\ldots,C_k$ of disjoint cycles in a graph $G$ is
  \emph{planarly nested} if for any $i\in \{1,\ldots,k\}$, there
  exists a $C_i$-component $H_i$, such that $H_1\supset \ldots \supset
  H_k$, and the graph obtained from $G$ by contracting to a single
  vertex all edges in the $C_i$-component $H_i$, except its feet, is
  planar.
\end{definition}

The following Lemma is due to Mohar \cite{Mohar_local_planarity}.  In
fact, \cite{Mohar_local_planarity} derives a somewhat more precise
bound for non-orientable drawings.
We will use a slightly weaker form which simplifies the notation, but
is still sufficient for our application.

\begin{lemma}[Mohar \cite{Mohar_local_planarity}]\label{lem:Mohar_planarly_nested}
  Let $G$ be a graph of Euler genus $g$.  Let $\phi$ be a drawing of
  $G$ into a surface ${\cal S}$ of Euler genus $g$.  Let
  $C_1,\ldots,C_k$ be a planarly nested sequence of cycles in $G$,
  where $k > g$.  Then, the cycles $\phi(C_1),\ldots,\phi(C_{k-g})$
  bound disks in ${\cal S}$ (see Figure \ref{fig:planarly_nested}).
\end{lemma}

\begin{figure}
\begin{center}
\iffull
\scalebox{0.6}{\includegraphics{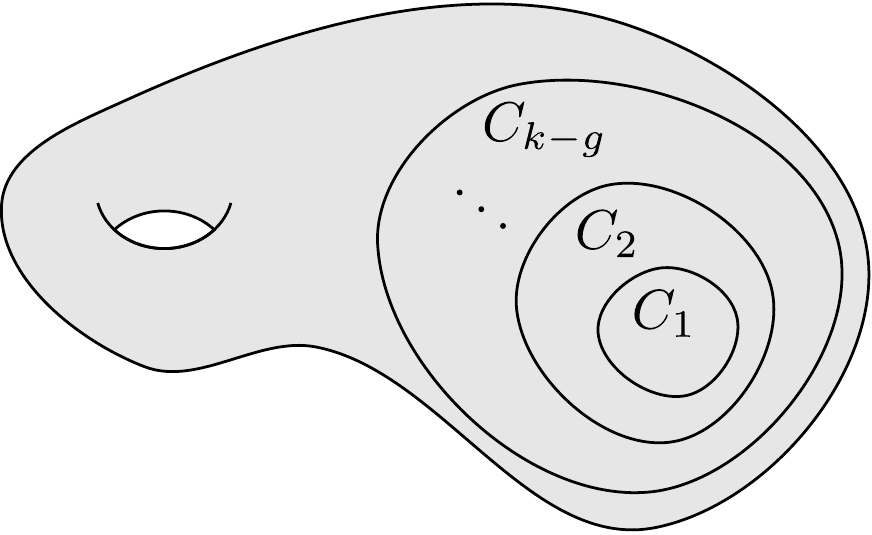}}
\fi
\ifabstract
\scalebox{0.6}{\includegraphics{figs/planarly_nested}}
\fi
\end{center}
\caption{
Planarly nested sequence for Lemma~\ref{lem:Mohar_planarly_nested}
\label{fig:planarly_nested}}
\end{figure}

The following crucial fact about planarly nested sequences, allows us to inductively maintain a graph of Euler genus $g$.

\begin{lemma}[Mohar \cite{Mohar_local_planarity}]\label{lem:planarly_nested_genus}
  Let $G$ be a graph of Euler genus $g$, and let $C_1,\ldots,C_k$ be a
  planarly nested sequence in $G$, with $k>g+1$.  Let $H$ be the
  $C_1$-component containing $C_2$, and let $G'=C_1\cup H$.  Then,
  $\eg(G')=\eg(G)=g$.
\end{lemma}

We can use our algorithm for computing a flat grid minor (lemma
\ref{lem:flat_grid}), to compute a planarly-nested sequence.  The next
lemma gives the precise statement.

\begin{lemma}[Computing a planarly nested sequence]\label{lem:computing_a_planarly_nested_sequence}
There exists a polynomial-time algorithm which
  given a graph $G$ of treewidth $t$, and maximum degree $\dmax$, and
  an integer $g\geq 1$,
  either correctly decides that $\eg(G)>g$, or it outputs a planarly
  nested sequence $C_1,\ldots,C_k$ in $G$, with
  $k=\Omega\left(\frac{t^{1/2}}{\dmax^{1/2} g^{7/2} \log^{15/4} n}\right)$.
\end{lemma}

\iffull
\begin{proof}
By Lemma \ref{lem:flat_grid}, we can compute in polynomial time a flat subgraph $H\subset G$, containing a $\left(r \times r\right)$-grid minor $M$, for some $r = \Omega\left(\frac{t^{1/2}}{\dmax^{1/2} g^{7/2} \log^{15/4} n}\right)$, together with a minor mapping $f$ for $M$.

Since $H$ is flat, we can compute a planar drawing $\phi$ of $H$, such that all the edges in $E(V(H), V(G)\setminus V(H))$ are incident to the unbounded face of $\phi$.
The drawing $\phi$ induces a planar drawing $\phi'$ of $M$.
There exist a sequence of at least $k=\lfloor r/4 \rfloor$ pairwise disjoint cycles $C_1',\ldots,C_k'$ in $M$, such that for any $i\in \{1,\ldots,k\}$, the cycle $\phi'(C_i')$ bounds a disk ${\cal D}_i'$ in the plane, for any $j\in \{1,\ldots,k-1\}$, $\phi'(C_j') \subset {\cal D}_{j+1}'$, and there exists $x\in V(G)$, such that $\phi'(x)$ lies in the interior of ${\cal D}_1$.
For each $i\in \{1,\ldots,k\}$, we can find a cycle $C_i$ in $H$, such that $V(C_i) \subseteq f(V(C_i'))$.
It is straightforward to check that the sequence of cycles $C_1,\ldots,C_k$ is planarly nested.
\end{proof}
\fi


The next lemma gives a technical condition that follows by the properties of normalized graphs, and will be used by our algorithm for computing a universal path.

\begin{lemma}[Uniqueness of the interior]\label{lem:unique_interior}
  Let $G$ be a connected normalized graph, and let $C,F \subseteq G$
  be vertex-disjoint cycles.  Let $\phi,\phi'$ be drawings of $G$ into
  surfaces ${\cal S}$, and ${\cal S}'$ respectively.  Suppose that
  there exist disks ${\cal D}_C, {\cal D}_F \subset {\cal S}$, and
  ${\cal D}'_C, {\cal D}'_F \subset {\cal S}'$, with ${\cal D}_C
  \subset {\cal D}_F$, and ${\cal D}'_C \subset {\cal D}'_F$, such
  that $\partial {\cal D}_C = \phi(C)$, $\partial {\cal D}_F =
  \phi(F)$, $\partial {\cal D}'_C = \phi'(C)$, and $\partial {\cal
    D}'_F = \phi'(F)$.  Then, the set of vertices (resp.~edges) that
  lie inside ${\cal D}_C$ in the drawing $\phi$, is the same as the
  set of vertices (resp.~edges) that lie inside ${\cal D}'_C$ in the
  drawing $\phi'$, i.e.
$\phi^{-1}({\cal D}_C \cap \phi(G)) = \phi'^{-1}({\cal D}'_C \cap \phi'(G))$.
\end{lemma}
\iffull
\begin{proof}
Let
\[
X = \phi^{-1}({\cal D}_C \cap \phi(G)),
\]
and 
\[
X' = \phi'^{-1}({\cal D}'_C \cap \phi'(G)).
\]
We need to show that $X=X'$.  Let $H$ be the $F$-component of $G$
containing $C$.  Since $G$ is connected, and $C$ separates both $X$,
and $X'$ from $F$, we have $X \subseteq H$, and $X'\subseteq H$.  Let
$G'=H\cup F$.  Let also $G''$ be the graph obtained from $G'$ by
adding a new vertex $f$, and connecting it to every vertex in $V(F)$.
Clearly, $G''$ is a planar graph.

Let $\Gamma''$ be the graph obtained from $G''$ by replacing all
maximal induced paths by edges.  We argue that $\Gamma''$ is
3-connected.  For the sake of contradiction suppose that there exists
a separator $Z\subset V(\Gamma'')$, with $|Z|\leq 2$.  Observe that if
$f\notin Z$, then some $Z$-component of $\Gamma''$ induces a free
subgraph in $G''$, which much also be a free in $G$, which contradicts
the fact that $G$ is normalized.  Therefore, we may assume that $f\in
Z$.  Let $Z=\{f,f'\}$.  Since $f$ is connected to every vertex in
$V(F)$, and $F$ is a connected subgraph, it follows that there exists
a component of $\Gamma'' \setminus Z$ that contains $F$.  This implies
however that $\{f'\}$ is also a separator in $\Gamma''$.  It follows
that some $\{f'\}$-component must be free in $\Gamma''$, and therefore
also in $G$, which again contradicts the fact that $G$ is normalized.
We have thus established that $G''$ is a subdivision of a 3-connected
graph.  The restrictions of $\phi$, and $\phi'$ on ${\cal D}_F$, and
${\cal D}_F'$ respectively, can be extended to drawings $\psi$, and
$\psi'$ of $G''$.  Since $G''$ is the subdivision of a planar graph,
it follows that it admits a combinatorially unique planar drawing.
This implies that $X=X'$, which concludes the proof.
\end{proof}
\fi

We are now ready to give our algorithm for computing a universal
patch.  This is the main result of this section.

\begin{proof}[\proofof Lemma \ref{lem:computing_a_universal_patch}]
  By Lemma \ref{lem:computing_a_planarly_nested_sequence} there exists
  a universal constant $\beta$, such that we can compute a planarly
  nested sequence $C_1,\ldots,C_k$ in $G$, for some $k \geq \beta
  \frac{t^{1/2}}{\dmax^{1/2} g^{7/2} \log^{15/4} n}$.  For a sufficiently large
  constant $\alpha>0$, we get $k\geq g+3$.  Let $\phi$ be a drawing of
  $G$ into a surface ${\cal S}$ of Euler genus $g$.  By Lemma
  \ref{lem:Mohar_planarly_nested} we have that the cycles
  $\phi(C_1),\phi(C_2),\phi(C_3)$ bounds disks ${\cal D}_1, {\cal
    D}_2, {\cal D}_3$ respectively, with
\[
{\cal D}_1 \subset {\cal D}_2 \subset {\cal D}_3.
\]

We now define a sequence of cycles $\Psi^{(0)},\ldots,\Psi^{(r)}$,
where $\Psi^{(0)}=C_2$, and $\Psi^{(r)}$ will be the desired cycle
$C$.  For every $i\in \{1,\ldots,r\}$, let $H^{(i)}$ be the connected
connected of $G \setminus \Psi^{(i)}$ containing $C_3$.  Let also
$W^{(i)} = \Psi^{(i)} \cup H^{(i)}$, and $X^{(i)} = G \setminus
W^{(i)}$.  Suppose we are given $\Psi^{(i)}$.  If the graph $W^{(i)}$
is normalized, then we set $r=i$, and therefore $C = \Psi^{(i)}$.
Otherwise, we proceed to define $\Psi^{(i+1)}$.  Since $W^{(i)}$ is
not normalized, this means that there exists a free vertex-induced
subgraph $Q^{(i)}$ of $W^{(i)}$.  Since $G$ is normalized, it follows
that $Q^{(i)}$ is not free in $G$.  Therefore, $V(Q^{(i)})\cap
V(\Psi^{(i)}) \neq \emptyset$.

We first argue that $Q^{(i)}$ cannot be a petal in $W^{(i)}$.  Suppose
that, to the contrary, $Q^{(i)}$ is a petal, with portal some vertex
$t$.  If $t\in V(\Psi^{(i)})$, then all edges between $V(Q^{(i)})$,
and $V(G\setminus Q^{(i)})$ in $G$, must be incident to $t$.  This
implies that $Q^{(i)}$ is also a vertex-induced subgraph of $G$, and
therefore also a petal in $G$, which contradicts the fact that $G$ is
normalized.  Therefore, we must have $t\in V(H^{(i)})$.  If
$V(Q^{(i)}) \cap V(\Psi^{(i)}) = \emptyset$, then clearly $Q^{(i)}$ is
a petal in $G$, which, again, is a contradiction.  Thus, it must be
$V(Q^{(i)}) \cap V(\Psi^{(i)}) \neq \emptyset$.  Since $\Psi^{(i)}$ is
a 2-connected subgraph, and $Q^{(i)}$ is a petal in $W^{(i)}$, it
follows that $\Psi^{(i)} \subseteq Q^{(i)}$.  Observe that
$\Psi^{(i)}$ bounds a disk in $\phi$, and therefore separates
$X^{(i)}$ in $G$.  This implies that $Q^{(i)} \cup X^{(i)}$ is a petal
in $G$, which again contradicts the fact that $G$ is normalized.  This
completes the proof that $Q^{(i)}$ cannot be a petal in $W^{(i)}$.

Therefore, $Q^{(i)}$ is a clump in $W^{(i)}$.  Let $t,t'$ be the
portals of $Q^{(i)}$.  If none of the portals are in $V(\Psi^{(i)})$,
then by arguing as above we deduce that $Q^{(i)}$ must be a clump in
$G$, which contradicts the fact that $G$ is normalized.  Thus, we may
assume $t \in V(\Psi^{(i)})$.  If $t'\notin V(\Psi^{(i)})$, then it
must be $t' \in V(H^{(i)})$.  But in this case, we again conclude that
$Q^{(i)}$ must be a clump in $G$, which is a contradiction.
Therefore, we must have $t,t'\in V(\Psi^{(i)})$.

We argue that $V(Q^{(i)})\cap V(C_3) = \emptyset$.  Suppose to the
contrary that $V(Q^{(i)})\cap V(C_3) \neq \emptyset$.  Since both
portals of $Q^{(i)}$ are in $V(\Psi^{(i)})$, it follows that $H^{(i)}
= Q^{(i)}$, which contradicts the fact that $Q^{(i)}$ is planar.  This
establishes that $V(Q^{(i)})\cap V(C_3) = \emptyset$.

Since $V(Q^{(i)})\cap V(C_3) = \emptyset$, this means that
$\phi(Q^{(i)}) \subset {\cal D}_3$.  Therefore, the drawing $\phi$,
induces a planar drawing $\psi^{(i)}$ of $Q^{(i)}$.  The outer face of
$\psi^{(i)}$ consists of two walks $L^{(i)}_1$, and $L^{(i)}_2$
between $t$, and $t'$.  Let $K^{(i)}_1$, $K^{(i)}_2$ be the two paths
in $\Psi^{(i)}$ between $t$, and $t'$.  There exists $l\in \{1,2\}$,
$k\in \{1,2\}$, such that the cycle $\phi(L^{(i)}_l \cup K^{(i)}_k)$
bounds a disk ${\cal F}^{(i)}\subset {\cal S}$, with
$\phi(L^{(i)}_{1-l} \cup K^{(i)}_{1-k}) \subset {\cal F}^{(i)}$.  We
can now define $\Psi^{(i+1)} = L^{(i)}_l \cup K^{(i)}_k$.
\begin{center}
\ifabstract
\scalebox{0.44}{\includegraphics{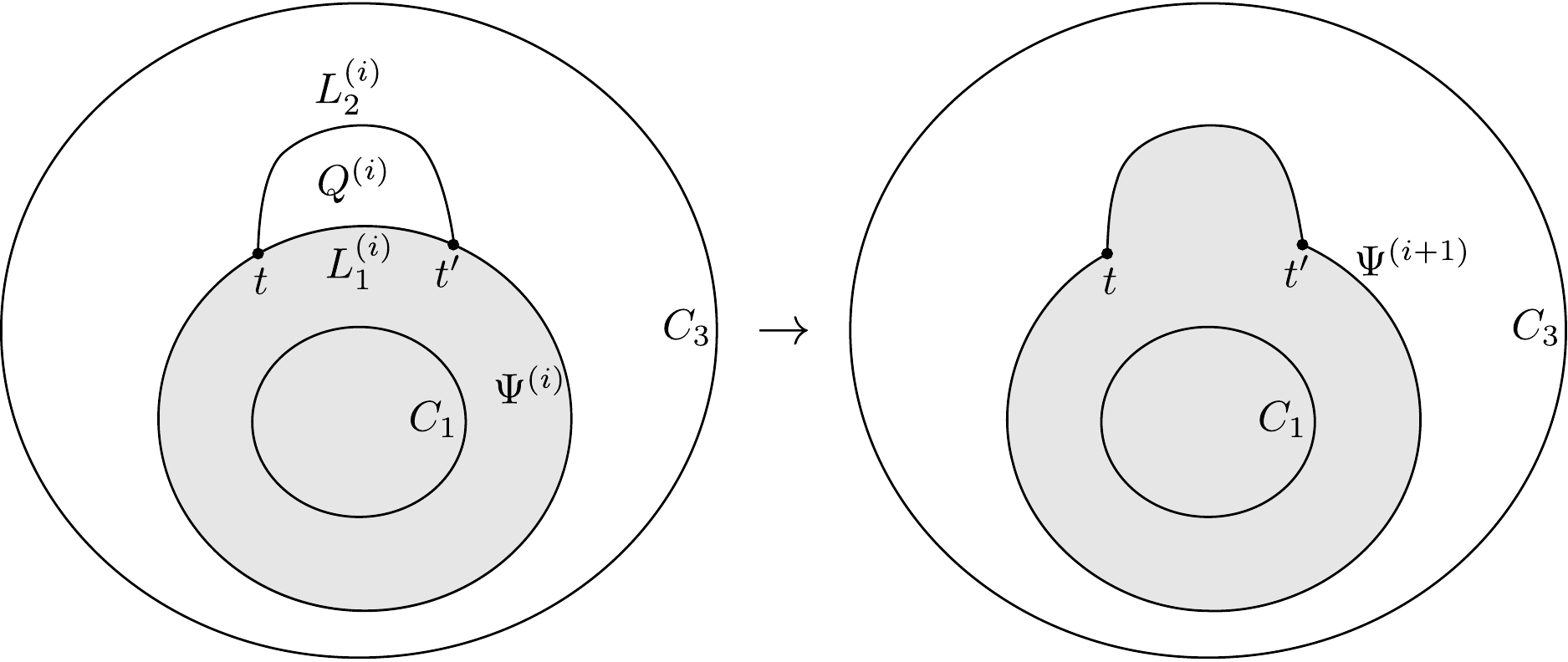}}
\fi
\iffull
\scalebox{0.6}{\includegraphics{figs/computing_a_universal_patch}}
\fi
\end{center}
This completes the definition of the sequence $\Psi^{(0)},\ldots,\Psi^{(t)}$.

We first show that the sequence terminates at a finite $t\geq 0$.
Assume $t>0$, since otherwise there is nothing to show.  For every
$i\in \{0,\ldots,t-1\}$, we have $X^{(i+1)} \supseteq X^{(i)}$.
Moreover, since $Q^{(i)}$ is free, it cannot consist of a single path
between its terminals.  Therefore, there exists at least one vertex $x
\in L^{(i)}_{1-l} \setminus L^{(i)}_l$.  Therefore, $x \in X^{(i+1)}$,
and $x\notin X^{(i)}$, which implies that $X^{(i+1)} \supsetneq
X^{(i)}$.  This implies that the sequence
$\Psi^{(0)},\Psi^{(1)},\ldots$, terminates at some $\Psi^{(t)}$, for
some finite $t\geq 0$.

We can now define $X = X^{(t)} \cup \Psi^{(t)}$.  We have already
shown that $G\setminus (X\setminus C)$ is normalized.  It remains to
show that $(X,C)$ is a universal patch in $G$.  Let $\phi'$ be a
drawing of $G$ into a surface ${\cal S}'$ of Euler genus $g$.  The
sequence of cycles $C_1,C,C_3,\ldots,C_k$ is planarly nested, and
therefore by Lemma \ref{lem:Mohar_planarly_nested} we have that the
cycle $\phi'(C)$ bounds a disk ${\cal D}'_C \subset {\cal S}'$.
Similarly, the cycle $\phi'(C_3)$ bounds a disk ${\cal D}_F \subset
{\cal S}'$.  By Lemma \ref{lem:unique_interior}, setting $F=C_3$, we
have that in any such drawing $\phi'$, the sets of vertices, and edges
of $G$ that are mapped inside ${\cal D}'_C$, are uniquely determined.
This proves condition (1) of Definition \ref{defn:universal_patch}.
Condition (2) follows directly by applying Lemma
\ref{lem:planarly_nested_genus} on the planarly nested sequence
$C,C_3,\ldots,C_k$.  Therefore, $(X,C)$ is a universal patch, as
required.  This concludes the proof.
\end{proof}

\iffull

\section{Approximating the orientable genus}\label{sec:orientable_genus}

In this section we give an algorithm for approximating the
\emph{orientable} genus of a graph.  The main idea is the following.
Let $G$ be the given graph of maximum degree $\dmax$, and let
$g=\genus(G)$.  We first use the algorithm from Theorem \ref{thm:main}
to compute a drawing of $G$ into a surface of Euler genus
$g'=\dmax^{O(1)} g^{O(1)} \log^{O(1)} n$.  If the resulting surface is
orientable, then we are done.
\footnote{
We remark that there is a simple polynomial-time algorithm for testing orientability.
Since we are dealing with cell-embeddings, this amounts to deciding whether there exists a way of orienting every face, so that the resulting orientations are consistent along edges.
This can be done by first orienting an arbitrary face, and then continuing to adjacent faces.
The algorithm terminates either when a globally-consistent orientation is found, certifying that the surface is orientable, or when we reach a face that cannot be consistently oriented (w.r.t.~its adjacent faces), certifying that the surface contains a M\"{o}bius band, and is therefore nonorientable.
}
Otherwise, we argue that there exists a
set of $g^{O(1)}$ vertices, whose removal decreases the genus of the
current surface.  More precisely, we show that if this is not the
case, then there exists a large ``M\"{o}bius'' grid minor, i.e.~a
graph ``densely'' embedded into the M\"{o}bius band (to be formally
defined later).  Such a minor has large orientable genus, which leads
to a contradiction.  After repeating this procedure at most $O(g')$
times, we arrive at a drawing of a subgraph $G'\subseteq G$, into some
orientable surface of genus at most $g'$.  We extend this drawing to a
drawing of $G$, simply by adding one handle for every removed edge.
This results into the desired drawing.

We begin by recalling some standard definitions, capturing the notion
of a ``dense'' embedding.  Let $G$ be a graph, and let $\phi$ be a
drawing of $G$ into a surface ${\cal S}$.  A \emph{noose} (in $\phi$)
is a loop in ${\cal S}$, meeting $\phi(G)$ only on $\phi(V(G))$.  The
\emph{length} of a noose $\gamma$ is defined to be
\[
\len(\gamma) = | \{v \in V(G) : \phi(v) \in \gamma\} |.
\]
The \emph{representativity} of $\phi$ is defined to be the smallest length of all noncontractible nooses in $\phi$.
In a similar vain, we also say that a curve $\xi$ with distinct endpoints, is a \emph{chain} (in $\phi$), if it meets $\phi(G)$ only on $\phi(V(G))$, and both its endpoints are in $\phi(V(G))$.
The \emph{length} of a chain $\xi$ is defined to be
\[
\len(\xi) = | \{v \in V(G) : \phi(v) \in \xi\} | - 1.
\]
We remark that the length of a chain is always non-negative, and it is zero if and only if the chain consists of a single point in $\phi(V(G))$.

The following result by Fiedler et al.~\cite{JGT:JGT3190200305}, gives
an obstruction to orientable genus, in terms of the representativity
of projective graphs.

\begin{lemma}[Fiedler et
  al.~\cite{JGT:JGT3190200305}]\label{lem:projective_genus}
  Let $G$ be a graph drawn into the projective plane, with
  representativity $\rho\neq 2$.  Then $\genus(G) = \left\lfloor
    \rho/2 \right\rfloor$.
\end{lemma}

We now define the graph which we will use as an obstruction to the
orientable genus of graphs drawn into a nonorientable surface.

\begin{definition}[M\"{o}bius grid]
  Let $k,l \geq 1$.  Let $G$ be the $(k\times l)$-grid.  For any $i\in
  \{1,\ldots,k\}$, $j\in \{1,\ldots,l\}$, let $v_{i,j}$ be the vertex
  at the $k$-th row, and $l$-th column of $G$.  Let $H$ be the graph
  obtained from $G$ by adding for every $j\in \{1,\ldots,l\}$, the
  edge $\{v_{1,l}, v_{k,l-j}\}$.  We call $H$ the \emph{$(k\times
    l)$-M\"{o}bius grid}.\footnote{The $(k \times l)$-M\"{o}bius grid is sometimes also called \emph{$(l \times l \times k)$-projective grid} (see \cite{DBLP:journals/jgt/Randby97})
}
\end{definition}

Lemma \ref{lem:projective_genus} implies the following.

\begin{corollary}\label{cor:mobius_genus}
  Let $k\geq 3$, and let $G$ be the $(k\times k)$-M\"{o}bius grid.
  Then, $\genus(G) = \Omega(k)$.
\end{corollary}

For two loops $\gamma, \delta$ in some surface, and an integer $t$, we
say that $\gamma$ is \emph{$t$-freely homotopic} to $\delta$, if
$\gamma$ is homotopic to $\delta^t$, where $\delta^t$ denotes the loop
obtained by concatenating $t$ copies of $\delta$.

The following result is due to Brunet et.~al
\cite{DBLP:journals/jct/BrunetMR96}.  The precise formulation cited
here, is implicit in their proof (proof of theorem 6.1 in
\cite{DBLP:journals/jct/BrunetMR96}, with inductive invariant in step
III.A).

\begin{lemma}[Brunet, Mohar, and Richter \cite{DBLP:journals/jct/BrunetMR96}]\label{lem:brunet_orientation-reversing}
  Let $G$ be a graph, and let $\phi$ be a drawing of $G$ into a
  nonorientable surface ${\cal S}$, with representativity $\rho$.  Let
  $\gamma$ be an orientation-reversing noose, of minimum length.
  Then, $G$ contains a set of $k=\lfloor (\rho-1)/4 \rfloor$ disjoint
  pairwise homotopic cycles $C_1,\ldots,C_k$, satisfying the following
  conditions:
\begin{description}
\item{(1)} For any $i\in \{1,\ldots,k\}$, we have $\phi(C_i)\cap
  \gamma = \emptyset$, and $\phi(C_i)$ is traversed by a loop 2-freely
  homotopic to $\gamma$.
\item{(2)} For any $i\in \{1,\ldots,k\}$, for every $v\in V(C_i)$,
  there exists $w\in V(G)$, with $\phi(w)\in \gamma$, and a chain
  $\xi$ with endpoints $\phi(v)$, and $\phi(w)$, such that $\len(\xi)
  \leq i+1$.
\end{description}
\end{lemma}


The following lemma is the main technical ingredient required by our
approximation algorithm of the orientable genus.  It shows that any
graph drawn into a nonorientable surface of nonorientable genus at
least 2, and with large enough representativity, contains a large
M\"{o}bius grid minor.  This fact appears to be understood by certain
experts.  We give a formal proof for completeness.  Perhaps something
less obvious is that we actually obtain a polynomial dependence of the
size of the M\"{o}bius grid minor, in terms of representativity.  This
polynomial dependence is necessary for our application.

\begin{lemma}\label{lem:mobius_grid_minor}
  Let $G$ be a graph, and let $\phi$ be a drawing of $G$ into a
  nonorientable surface ${\cal S}$, of nonorientable genus at least
  $2$.  Let $\rho>10$ be the representativity of $\phi$.  Then, $G$
  contains a $(r \times r)$-M\"{o}bius grid minor, for some
  $r=\Omega(\sqrt{\rho})$.
\end{lemma}
\begin{proof}
  Let $\gamma$ be an orientation-reversing loop in ${\cal S}$, of
  minimum length (w.r.to $\phi$).  Since $\gamma$ is a noose, we have
  $\rho' = \len(\gamma)\geq \rho$.  Let $C_1,\ldots,C_k$ be the
  collection of cycles given by lemma
  \ref{lem:brunet_orientation-reversing}, with $k =
  \lfloor(\rho-1)/4\rfloor$.  There exists a neighborhood $M$ of
  $\gamma$, homeomorphic to the M\"{o}bius strip, such that for any
  $i\in \{1,\ldots,k\}$, $\phi(C_i)\subset M$.  Cutting $M$ along
  $\gamma$, we obtain a graph $G'$ drawn into a cylinder $M'$.  Let
  $\phi'$ be the induced drawing of $G'$ on $M'$.  
  One of the boundaries of the cylinder $M'$, is
\[
\gamma' = \gamma^{(1)} \cup \gamma^{(2)},
\]
where for each $\{1,2\}$, $\gamma^{(i)}$ is a segment corresponding to a copy of $\gamma$.
Let
\[
U = \{v\in V(G) : \phi(v) \in \gamma\}.
\]
For any $v\in U$, there exist two copies of $v$ in $V(G')$; for any $i\in \{1,2\}$, let $v^{(i)}$ be the copy of $v$ with $\phi'(v^{(i)}) \in \gamma^{(i)}$.
Let also
\[
U' = \{v^{(i)}: v\in U \mbox{ and } i\in \{1,2\} \}.
\]

\begin{center}
\scalebox{0.8}{\includegraphics{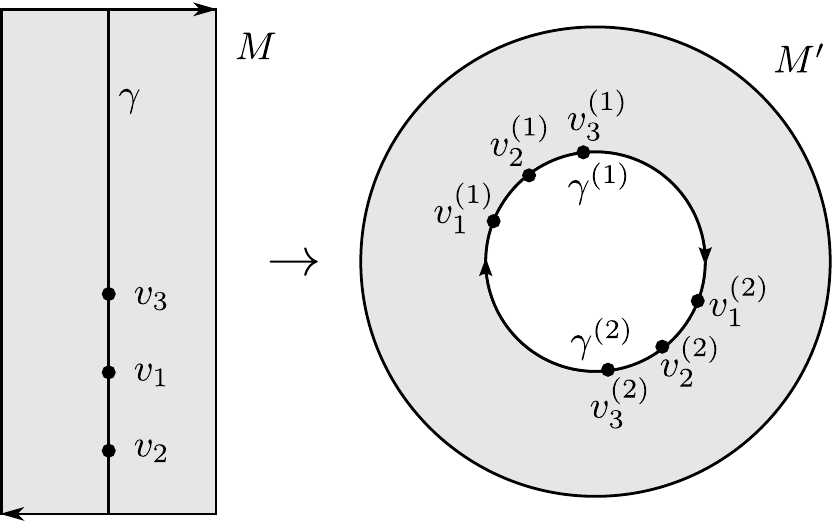}}
\end{center}

Let $t_1 = \sqrt{\rho}/10$.  For every $w\in V(C_{t_1})$, there exists
$v\in U$, and a chain $\xi_{w}$ in $\phi$, between $\phi(w)$, and
$\phi(v)$, with $\len(\xi_w) \leq t_1+1$.  By possibly shortcutting
$\xi_w$, we can assume that it intersects $\gamma$ only on $\phi(v)$.
Moreover, the collection of chains $\{\xi_w\}_{w\in V(C_{t_1})}$ can
be chosen so that no two curves intersect transversely.  This can be
done by observing that the length of chains induces a metric on
$V(G)$, and by choosing the curves $\xi_w$ to be shortest possible,
breaking ties in a consistent fashion\footnote{For example, we can
  redefine the length of a chain $\xi$ to be the total \emph{weighted
    cost} of the faces traversed by $\xi$, counting multiplicities,
  and where the weight of a face is set to $1+\eps$, for some
  perturbation $\eps=O(1/n^2)$.}.

Clearly, for every $w\in V(C_{t_1})$, we have $\xi_w\subset M$, since
otherwise $\xi_w$ has to intersect all the cycles
$C_{{t_1}+1},\ldots,C_k$, and therefore its length cannot be at most
$t_1+1$.  Since $\xi_w\subset M$, and $\xi_w$ intersects $\gamma$ only
on an endpoint $\phi(v)$, for some $v\in U$, it follows that $\xi_w$
lifts to a chain $\xi_w'\subset M'$, in $\phi'$, with endpoints
$\phi'(w)$, and $\phi'(v^{(j)})$, for some $j\in \{1,2\}$.

Let $W$ be a maximal subset $W\subseteq V(C_t)$, such that for any $w \neq w'\in W$, we have $\xi_w' \cap \xi_{w'}' = \emptyset$.
Let $t_2$ be the smallest integer such that for any set $A\subset U'$
consisting of $t_2$ consecutive vertices along the loop $\gamma'$,
there exists some $w\in W$, with $\xi_w'$ having an endpoint in $A$.

If $t_2 \leq \sqrt{\rho}$, then we can find the desired M\"{o}bius
grid minor as follows.  Let $t_3=\lfloor \rho'/(2 t_2) \rfloor$.  Let
$v_1^{(1)},\ldots,v_{\rho'}^{(1)},v_1^{(2)},\ldots,v_{\rho'}^{(2)}$ be
an ordering of the vertices in $U'$ induced by a traversal of
$\gamma'$.  For any $i\in \{1,\ldots,t_3\}$, and $j\in \{1,2\}$, let
\[
A_i^{(j)} = \{v_{(i-1)t_3 + 1}^{(j)}, v_{(i-1)t_3 + 2}^{(j)}, \ldots, v_{i t_3}^{(j)}\}.
\]
By planarity of $M$ (see e.g.~\cite{DBLP:conf/focs/ChekuriKS04}), we
obtain a $(t_1 \times 2 t_3)$-grid minor $H'$ in $G'$, with minor
mapping $f:V(H') \to 2^{V(G')}$, such that the vertices in the top row
of $H'$ are
\[
f^{-1}(A_1^{(1)}),\ldots,f^{-1}(A_{t_3}^{(1)}), f^{-1}(A_1^{(2)}),\ldots,f^{-1}(A_{t_3}^{(2)}),
\]
and in this order.  Pulling $f$ back to $G$, we arrive at a minor $H$
in $G$, where $H$ is obtained from $H'$ by identifying for every $i\in
\{1,\ldots,t_3\}$, the pair of vertices $f^{-1}(A_i^{(1)})$, and
$f^{-1}(A_i^{(2)})$ of $H'$.  It is straightforward to check that $H$
contains a $(t_4 \times t_4)$-M\"{o}bius grid minor, where $t_4 \geq
\min\{t_1/2, t_3/2\}-1 = \Omega(\min\{\sqrt{\rho}, \rho'/t_2\}) =
\Omega(\min\{\sqrt{\rho}, \rho/t_2\})$.  Therefore, if $t_2 \leq
\sqrt{\rho}$, then we are done.

It remains to consider the case $t_2 > \sqrt{\rho}$.  We consider the
following two subcases:
\begin{description}
\item{Case 1: $|W| \geq 2$.}  Pick $w \neq w' \in W$, such that the
  clockwise segment $P$ of $C_{t_1}$ between $w$, and $w'$, does not
  contain any other vertex $w''\in W \setminus \{w,w'\}$.  Let
  $\phi'(v)$, and $\phi'(v')$, be the endpoints of $\xi_{w}'$, and
  $\xi_{w'}'$ respectively in $\gamma'$, for some $v,v'\in U'$.  By
  the assumption, we can pick $w$, and $w'$, so that the clockwise
  distance between $v$, and $v'$ along $\gamma'$ is at least $t_2$.
  Let $P = w_1,\ldots,w_{\ell}$, be that segment, where $w_1=w$, and
  $w_{\ell} = w'$.  By the maximality of $W$, it follows that for
  every $i\in \{1,\ldots,\ell\}$, the chain $\xi_{i}'$ intersects
  either $\xi_w'$, or $\xi_{w'}'$ (or both).  By planarity, it follows
  that there exists $s\in \{1,\ldots,\ell-1\}$, such that for any
  $i\in \{1,\ldots,s\}$, we have
\[
\xi_{w_i}'\cap \xi_{w}' \neq \emptyset,
\]
and for any $j\in \{s+1,\ldots,\ell\}$, we have
\[
\xi_{w_j}'\cap \xi_{w'}' \neq \emptyset,
\]
For any $i\in \{1,\ldots,s\}$, let $\zeta_{w_i}$ be the segment of
$\xi_{w_i}'$ between $\phi'(w_i)$, and the first intersection point
with $\xi_{w}'$.  Similarly, for any $j\in \{s+1,\ldots,\ell\}$, let
$\zeta_{w_j}$ be the segment of $\xi_{w_j}'$ between $\phi'(w_j)$, and
the first intersection point with $\xi_{w'}'$.  Let $\sigma$ be the
subcurve of $\xi_w'$ between $\phi'(v)$, and the point $\xi_w' \cap
\zeta_{w_s}$.  Similarly, let $\sigma'$ be the subcurve of $\xi_{w'}'$
between $\phi'(v')$, and the point $\xi_{w'}' \cap \zeta_{w_{s+1}}$.
Since $\{w_s,w_{s+1}\}\in E(G')$, we can pick a chain $\alpha$ between
$\phi'(w_s)$, and $\phi'(w_{s+1})$, with $\len(\alpha) = 1$.  Let
$\tau'$ be the curve obtained by the concatenation
\[
\tau' = \sigma \circ \zeta_{w_s} \circ \alpha \circ \zeta_{w_{s+1}} \circ \sigma',
\]
where $\circ$ denotes the usual concatenation of curves with a common endpoint.
\begin{center}
\scalebox{0.7}{\includegraphics{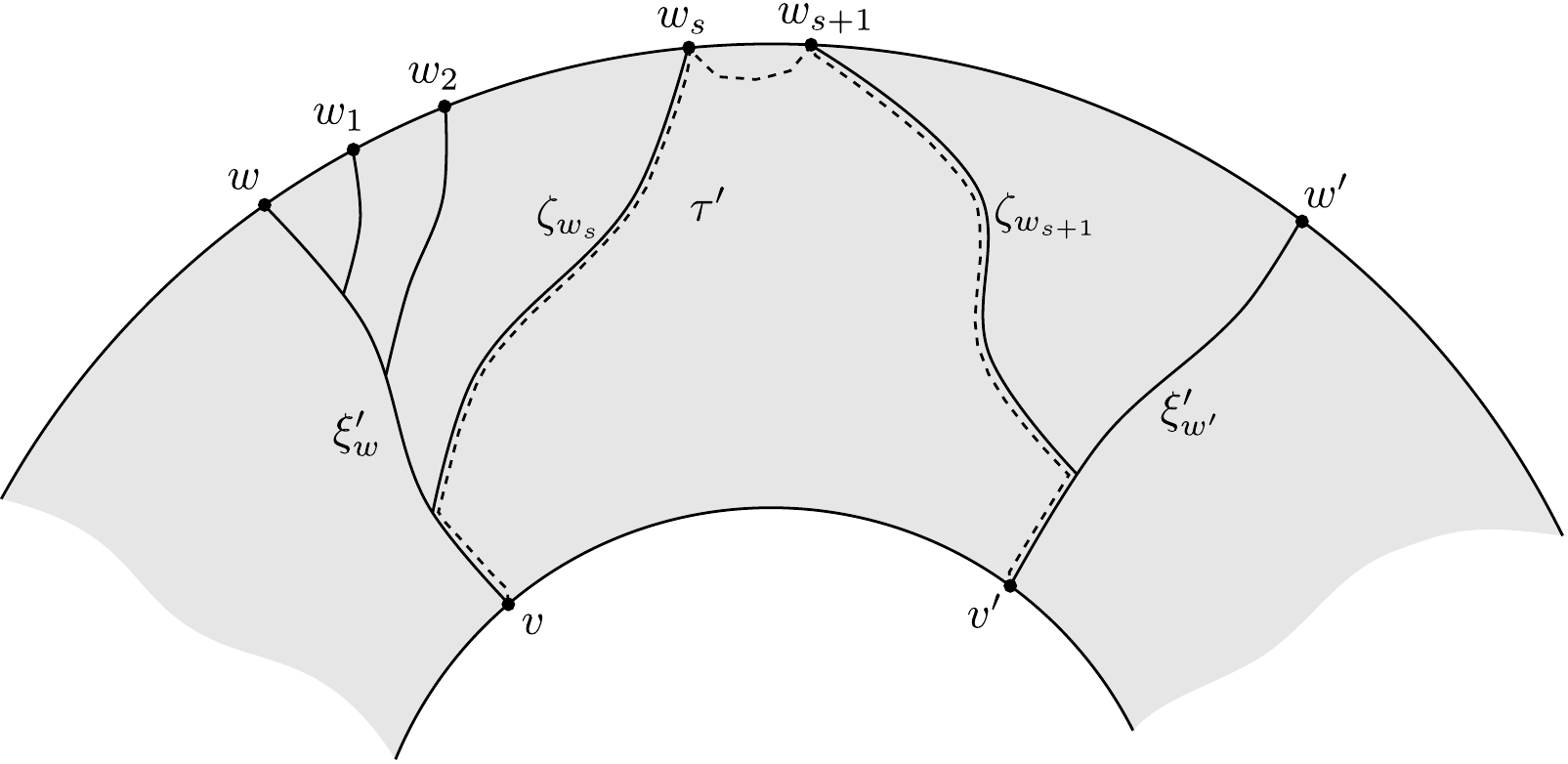}}
\end{center}
We have 
\begin{align*}
\len(\tau') &= \len(\sigma) + \len(\zeta_{w_s}) + \len(\alpha) + \len(\zeta_{w_{s+1}}) + \len(\sigma')\\
 &\leq \len(\xi_w') + \len(\xi_{w_s}') + \len(\alpha) + \len(\xi_{w_{s+1}}') + \len(\xi_{w'}')\\
 &\leq 4t_1 + 5.
\end{align*}

Let $v=v_i^{(j)}$, and $v'=v_{i'}^{(j')}$, for some $i,i'\in \{1,\ldots,\rho'\}$, and $j,j'\in \{1,2\}$.
Let $\delta_1$, and $\delta_2$ be the two arcs of $\gamma$ between $v_i$, and $v_{i'}$.
Let $\alpha$ be the clockwise distance between $v$, and $v'$, along $\gamma'$.
We have the following two subcases:
\begin{description}
\item{Case 1.1: Assume $\alpha < \rho'$.}
In this case, $\tau'$ corresponds to a chain $\tau$ in $\phi$.
One of the two loops $\tau \cup \delta_1$, or $\tau \cup \delta_2$, is a noose homotopic to $\gamma$, of length $\len(\tau)+\rho' - t_2 < \rho'$, which contradicts the fact that $\gamma$ is the shortest orientation-reversing noose.
\begin{center}
\scalebox{0.8}{\includegraphics{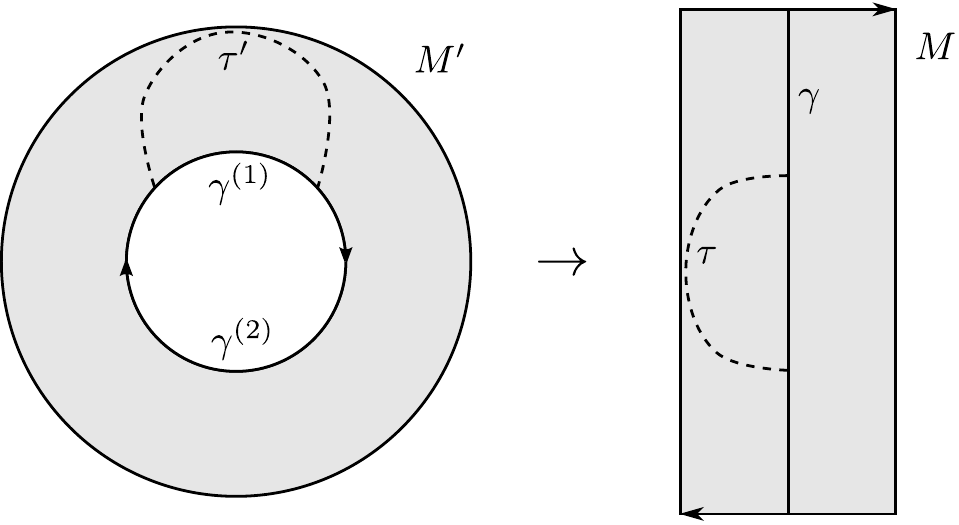}}
\end{center}

\item{Case 1.2: Assume $\alpha \geq \rho'$.}
We have that $\tau'$ corresponds to a chain $\tau$ in $\phi$, of length $\len(\tau) = \len(\tau')$.
Consider the loops $\eta_1 = \tau \cup \delta_1$, and $\eta_2 = \tau \cup \delta_2$.
Since ${\cal S}$ has nonorientable genus at least 2, it follows that both $\eta_1$, and $\eta_2$, are noncontractible (see e.g.~\cite{DBLP:journals/jct/BrunetMR96}).
One of these loops, say $\eta_1$, is an orientation-reversing noose, with $\len(\eta_1)=\len(\tau)+\alpha-\rho'$,
and $\eta_2$ is a non-contractible noose, which is 2-freely homotopic to $\gamma$, with $\len(\eta_2)=\len(\tau)+2\rho'-\alpha$.
If $\alpha\leq 2\rho'+\len(\tau)-\rho$, then $\len(\eta_1) \leq \rho' + 2\len(\tau) - \rho < \rho'$, which contradicts the fact that $\gamma$ is a shortest orientation-reversing noose in $\phi$.
Otherwise, if $\alpha < 2\rho' + \len(\tau)-\rho$, then $\len(\eta_2) < \rho$, which contradicts the fact that the representativity of $\phi$ is $\rho$. 
\begin{center}
\scalebox{0.8}{\includegraphics{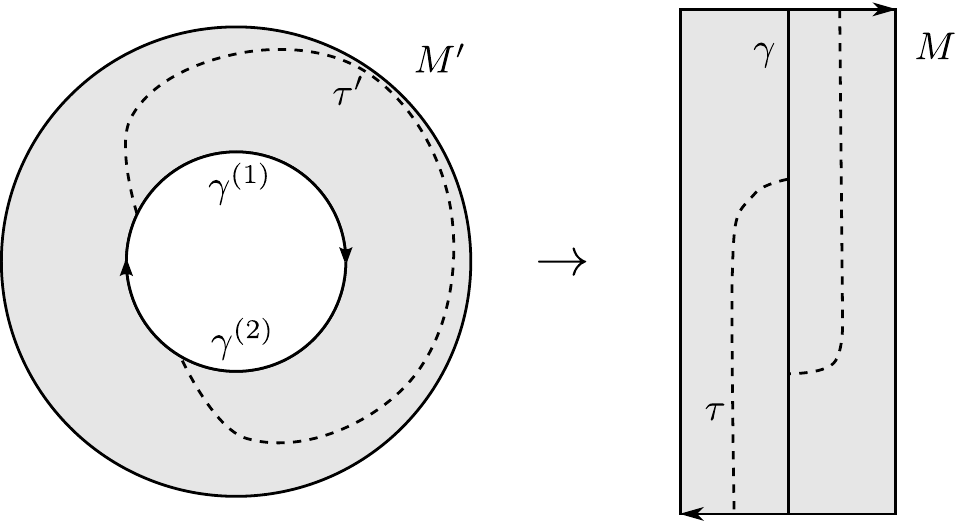}}
\end{center}
\end{description}

\item{Case 2: $|W| = 1$.}
Arguing as above, we construct a noncontractible noose in $\phi$, of length at most $4t_1+4 < \rho$, which is a contradiction.
\end{description}
This concludes the proof.
\end{proof}

We are now ready to prove the main result of this section.

\begin{theorem}[Approximating the orientable genus]\label{thm:orientable_genus}
There exists a polynomial-time algorithm which given a graph $G$ of maximum degree $\dmax$, and an integer $g>0$, either correctly decides that $\genus(G)>g$, or outputs a drawing of $G$ into a surface of orientable genus $O(\dmax^3 g^{14} \log^{19/2} n)$.
\end{theorem}

\begin{proof}
Using the algorithm from theorem \ref{thm:main}, we can either correctly decide that $\genus(G) > g$, or compute a drawing $\phi$ of $G$ into a surface ${\cal S}$ of Euler genus at most $g'=O(\dmax^2 g^{12} \log^{19/2} n)$.

We inductively compute a sequence of graphs $G_0,\ldots,G_t$, and for each $i\in \{1,\ldots,t\}$, we also compute a drawing $\phi_i$ of $G_i$ into some surface ${\cal S}_i$.
Set $G_0=G$, ${\cal S}_0={\cal S}$, and $\phi_0=\phi$.
Given $G_i$, ${\cal S}_i$, and $\phi_i$, if ${\cal S}_i$ is orientable, then we set $t=i$, and terminate the sequence at $G_i$.
Otherwise, we compute $G_{i+1}$, ${\cal S}_{i+1}$, and $\phi_{i+1}$ as follows.
Let $\rho_i$ be the representativity of $\phi_i$.
Let us assume first that $\rho_i > \alpha g^2$, for some constant $\alpha>0$, to be determined later.
If ${\cal S}_i$ is the real projective plane, then Lemma \ref{lem:projective_genus} implies that $\genus(G) > \lfloor \alpha g^2 / 2 \rfloor > g$.
In this case, we can therefore correctly decide that $\genus(G) > g$, for $\alpha \geq 4$.
Otherwise, we may assume that ${\cal S}_i$ has nonorientable genus at least 2.
In this case, we can apply Lemma \ref{lem:mobius_grid_minor} to obtain a $(r\times r)$-M\"{o}bius grid minor, for some $r \geq \beta \sqrt{\alpha g^2}$, for some universal constant $\beta>0$ (determined by the hidden constant in Lemma \ref{lem:mobius_grid_minor}).
By corollary \ref{cor:mobius_genus}, we have that $\genus(G) > \beta' \beta \sqrt{\alpha} g$, for some universal constant $\beta'>0$.
By setting $\alpha = 1/(\beta \beta')$, we correctly decide that, in this case, $\genus(G)>g$.
Finally, it remains to consider the case $\rho_i \leq \alpha g^2$.
We find a shortest noose $\gamma_i$ in $\phi_i$ (this can be done in polynomial time, e.g.~with the algorithm of Cabello et al.~\cite{DBLP:conf/compgeom/CabelloVL10a}).
We set $G_{i+1} = G_i \setminus X_i$, where $X_i$ is the set of vertices that $\gamma_i$ intersects.
We also set ${\cal S}_{i+1}$ be the surface obtained by cutting ${\cal S}_i$ along $\gamma_i$, and replacing any created punctures by disks.
Note that the resulting surface might be disconnected, but this does not affect our algorithm.
Let also $\phi_{i+1}$ be the induced drawing of $G_{i+1}$ into ${\cal S}_{i+1}$.

Observe that since we arrive at ${\cal S}_{i+1}$ by cutting ${\cal S}_i$ along a noncontractible curve, it follows that the sequence $G_1,\ldots,G_t$ terminates at some finite $t=O(g')$.
Let 
\[
X = \bigcup_{i=0}^{t-1} X_i,
\]
i.e.~$X = V(G) \setminus V(G_t)$ is the set of all vertices removed by the algorithm.
We have
\begin{align*}
|X| &= \sum_{i=0}^{t-1} \rho_i\\
 &\leq t \alpha g^2\\
 &= O(\dmax^2 g^{14} \log^{19/2} n).
\end{align*}
Since ${\cal S}_t$ is orientable, we can obtain a drawing of $G$ into
an orientable surface ${\cal S}'$ by adding for every edge incident to
a vertex in $X$, a new handle in ${\cal S}_t$.  This results into a
drawing of $G$ into an orientable surface ${\cal S}'$ of orientable
genus
\begin{align*}
\genus({\cal S}') &= \genus({\cal S}_t) + \dmax \cdot |X|\\
 &= O(\genus({\cal S}) + \dmax \cdot |X|)\\
 &= O(\dmax^3 g^{14} \log^{19/2} n),
\end{align*}
as required.
\end{proof}

\section{Approximating the crossing number}\label{sec:crossing_number}

In this section we give an algorithm for approximating the crossing
number of a graph by building on our algorithm for approximating the
orientable genus of a graph.  
We remark that we can obtain a similar algorithm for crossing number  by instead using our algorithm for Euler genus.
However, using orientable drawings simplifies the analysis.

\begin{definition}[Augmented cylinder]
Let $r,t\geq 1$, and let $G$ be the $(r \times t)$-cylinder.
Let $G'$ be the graph obtained from $G$ by adding the edge $\{u,v\}$, between some vertex $u$ in the top of $G$, and some vertex $v$ in the bottom of $G$.
We say that $G'$ is an \emph{$(r \times t)$-augmented cylinder}.
\end{definition}

The following fact is not difficult, it can be proven e.g.~using the arguments from \cite{ChuzhoyMS11}.

\begin{lemma}\label{lem:augmented_cylinder_crossing_number}
Let $r\geq 3$, and let $G$ be an $(r\times r)$-augmented cylinder.
Then, the crossing number of $G$ is $\Omega(r)$.
\end{lemma}

We remark that the same bound holds on the crossing number of an $(r \times 3)$-augmented cylinder.
This refined statement however, does not make any difference in our application.

We now show that any graph that admits a ``dense enough'' drawing into an orientable surface of positive orientable genus, contains a large augmented cylinder minor.
It seems that one should also be able to get a large toroidal minor, and in fact such a result has been obtained for graphs drawn into the torus by Schrijver \cite{DBLP:journals/jct/Schrijver93a}, and by Hlin\v{e}n\'{y}, and Salazar \cite{DBLP:conf/isaac/HlinenyS07}.
It is possible that an extension of these results to higher genus surfaces can be obtained by the techniques in \cite{DBLP:journals/jct/BrunetMR96}.
This however does not seem immediate.
Let us also point out that such an improvement would immediately yield an improved constant in the exponent of the approximation guarantee of our algorithm.

We will use the following results to construct the desired augmented cylinder.

\begin{lemma}[Brunet et al.~\cite{DBLP:journals/jct/BrunetMR96}]\label{lem:brunet_1}
Let $G$ be a graph, and let $\phi$ be a drawing of $G$ into an orientable surface ${\cal S}$, of positive orientable genus.
Let $\rho$ be the representativity of $\phi$.
Then, $G$ contains a set of $\lfloor (\rho-1)/2 \rfloor$ pairwise disjoint, pairwise freely homotopic nonseparating cycles.
\end{lemma}

\begin{lemma}[Brunet et al.~\cite{DBLP:journals/jct/BrunetMR96}]\label{lem:brunet_2}
  Let $G$ be a graph, and let $\phi$ be a drawing of $G$ into a
  surface, with representativity $\rho$.  Let $C$, $C'$ be disjoint
  homotopic noncontractible cycles in $G$.  Then, $G$ contains $\rho$
  pairwise disjoint paths, each contained within the cylinder bounded
  by $C$, and $C'$, and each having one end in $C$, and the other in
  $C'$.
\end{lemma}

We point out that Lemma \ref{lem:brunet_2} is a consequence of the  max-flow/min-cut duality.
We now prove that large representativity in a orientable surface
implies the existence of a large augmented cylinder minor.

\begin{lemma}\label{lem:augmented_cylinder}
  Let $G$ be a graph, and let $\phi$ be a drawing of $G$ into an
  orientable surface ${\cal S}$, with orientable genus $g>0$.  Let
  $\rho \geq 3$ be the representativity of $\phi$.  Then, $G$ contains
  a $(r\times r)$-augmented cylinder minor, for some $r =
  \Omega(\rho)$.
\end{lemma}
\begin{proof}
  By Lemma \ref{lem:brunet_1}, there exists a collection
  $C_1,\ldots,C_k$ of pairwise disjoint, pairwise freely homotopic
  nonseparating cycles in $G$, with $k = \Omega(\rho)$.  By a possible
  reordering of the indices, we have that there exists a cylinder
  ${\cal C} \subset {\cal S}$, with boundaries $C_1$, and $C_k$.  By
  Lemma \ref{lem:brunet_2}, there exists a collection
  $P_1,\ldots,P_{\ell}$ of pairwise disjoint paths in ${\cal C}$, each
  having an endpoint in $C_1$, and an endpoint in $C_k$, with $\ell =
  \lfloor (\rho-1)/2 \rfloor$.  
  We next argue that the graph $H=\left(\bigcup_{i=1}^{k}
    C_i\right) \cup \left( \bigcup_{j=1}^{\ell} P_j \right)$ contains
  a $(r\times r)$-cylinder minor, for some $r=\Omega(\rho)$.
  Construct a torus ${\cal T}$ by glueing a cylinder ${\cal C}'$ to ${\cal C}$ (i.e., each boundary cycle of ${\cal C}'$ is identified with a boundary cycle of ${\cal C}$).
  For each $j\in \{1,\ldots,\ell\}$, add an edge $e_j$ between the two endpoints of $P_j$, and embed $e_j$ into the cylinder ${\cal C}'$.
  This results into a drawing of the graph $H' = H \cup \left(\bigcup_{j=1}^{\ell} e_j\right)$ into the torus ${\cal T}$, with representativity $\rho'\geq \min\{k, \ell\}$.
  By the arguments in \cite{DBLP:conf/isaac/HlinenyS07} we can find a toroidal minor in $H'$, which induces a $(\Omega(\rho') \times \Omega(\rho'))$-cylinder minor in $H$.
   Since the cycles $C_1,\ldots,C_k$ are nonseparating, it follows that ${\cal S}\setminus {\cal C}$ is a connected surface.  This implies that there exists a path $Q$ in $G$ between $C_1$, and $C_k$, that
  intersects ${\cal C}$ only on its endpoints.  Contracting $Q$ into
  an edge, gives a $(k \times \ell)$-augmented cylinder minor, as
  required.
\end{proof}

We also need the following result due to Garcia-Moreno, and Salazar \cite{DBLP:journals/jgt/Garcia-MorenoS01}, which allows us to bound the crossing number of a graph $G$, in terms of the crossing number of a minor of $G$, with maximum degree 4.

\begin{lemma}[Garcia-Moreno, and Salazar \cite{DBLP:journals/jgt/Garcia-MorenoS01}]\label{lem:crossing_degree_four}
Let $G$ be a graph, and let $H$ be a minor in $G$, of maximum degree 4.
Then, $\crossingnumber(G) \geq \crossingnumber(H)/4$.
\end{lemma}

The following result allows us to reduce the problem of approximating
the crossing number, to the problem of computing a set of edges of
minimum cardinality, whose removal leaves a planar graph.  The first
result of this type was obtained in \cite{ChuzhoyMS11}, but the
following gives a slightly better dependence on the maximum degree.


\begin{lemma}[Chimani, and Hlinen{\'y} \cite{ChimaniH11}]\label{lem:crossing_number_edge_planarization}
  There exists a polynomial-time algorithm which given a graph $G$ of
  maximum degree $\dmax$, and a set $E^*\subseteq E(G)$, with
  $|E^*|=\ell$, such that $G\setminus E^*$ is a connected planar
  graph, outputs a drawing of $G$ into the plane, with at most
  $O(\dmax \cdot \ell \cdot (\ell+\crossingnumber(G)))$ crossings.
\end{lemma}

We are now ready to present our approximation algorithm for crossing number.

\begin{theorem}\label{thm:crossing_number}
  There exists a polynomial-time algorithm which given a connected
  graph $G$ of maximum degree $\dmax$, and an integer $k\geq 0$,
  either correctly decides that $\crossingnumber(G)>k$, or outputs a
  drawing of $G$ into the plane with at most $O(\dmax^{9} k^{30}
  \log^{19} n)$ crossings.
\end{theorem}
\begin{proof}
  If $G$ is planar, then there is nothing to be done.  We may
  therefore assume that $G$ is not planar. 

  Recall that $\genus(G) \leq \crossingnumber(G)$.  Therefore, using
  the algorithm from Theorem \ref{thm:orientable_genus}, we can either
  correctly decide that $\crossingnumber(G) > k$, or we can compute a
  drawing $\phi$ of $G$ into an orientable surface ${\cal S}$ of
  orientable genus $g=O(\dmax^3 k^{14} \log^{19/2} n)$.  We now proceed
  to inductively compute a sequence of subgraphs $G=G_0 \supset \ldots
  \supset G_t$.  For every $i\in \{0,\ldots,t\}$, we also compute a
  drawing $\phi_i$ of $G_i$ into some orientable surface ${\cal S}_i$.
  Set $\phi_0=\phi$, and ${\cal S}_0={\cal S}$.

  Given $G_i$, $\phi_i$, and ${\cal S}_i$, we proceed as follows.  If
  $G_i$ is planar, then we set $t=i$, and we terminate the sequence at
  $G_i$.  Otherwise, let $\rho_i$ be the representativity of $\phi_i$.
  We find a shortest nonseparating noose $\gamma_i$ in $\phi_i$ (using
  the algorithm from \cite{DBLP:conf/compgeom/CabelloVL10a}).  If
  $\len(\gamma_i) = \rho_i > \alpha k$, for some universal constant
  $\alpha>0$ (to be determined later), then by Lemma
  \ref{lem:augmented_cylinder} we find a $(r\times r)$-augmented
  cylinder minor in $G_i\subseteq G$, with $r>\beta \alpha k$, for
  some universal constant $\beta$.  Combining lemmas
  \ref{lem:augmented_cylinder_crossing_number}, and
  \ref{lem:crossing_degree_four}, we obtain that $\crossingnumber(G)
  \geq \crossingnumber(G_i) > \beta' \beta \alpha k$, for some
  universal constant $\beta'>0$.  By setting $\alpha=1/(\beta
  \beta')$, this leads to a proof that $\crossingnumber(G)>k$.
  Therefore, it remains to consider the case $\len(\gamma_i) \leq
  \alpha k$.  In this case, we set $G_{i+1} = G_i \setminus X_i$,
  where $X_i$ is the set of vertices in $\gamma_i$.  We also set
  ${\cal S}_{i+1}$ to be the surface obtained form ${\cal S}_i$ by
  cutting along $\gamma_i$, and placing disks on any possible created
  punctures.  Let also $\phi_{i+1}$ be the induced drawing of
  $G_{i+1}$ into ${\cal S}_{i+1}$.  This completes the definition of
  the sequence $G_0,\ldots,G_t$.  We argue that the sequence
  terminates at some finite $t$.  Indeed, we can only cut along at
  most $O(g)$ nooses before a surface becomes trivial (cutting along a
  separating noose disconnects the surface, but this does not affect
  our argument).  Therefore, $t = O(g)$.

  Let $X = \bigcup_{i=0}^{t-1} X_i$, and let $E^*$ be the set of all
  edges incident to vertices in $X$.  We have
\begin{align*}
|E^*| &\leq \dmax \cdot |X|\\
 &= \dmax \cdot \sum_{i=1}^t \rho_i\\
 &\leq \dmax t \alpha k\\
 &= O(\dmax^4 k^{15} \log^{19/2} n).
\end{align*}

Let $E'$ be a maximal subset of $E^*$ such that $G\setminus E'$ is
planar.  We claim that $G\setminus E'$ is a connected graph. Suppose
not, then since $G$ is connected, there must two connected components
$H_1,H_2$ of $G\setminus E'$ and an edge $uv \in E'$ such that $u \in
V(H_1)$ and $v \in V(H_2)$. Both $H_1$ and $H_2$ are planar and we can
draw them such that $u$ and $v$ are on the outer face of their
drawings respectively.  This allows us to add the edge $uv$ to $H_1
\cup H_2$ while maintaining planarity, thus contradicting the
maximality of $E'$.  Since $G\setminus E'$ is a planar connected
graph, using the algorithm from Lemma
\ref{lem:crossing_number_edge_planarization}, we obtain in polynomial
time a drawing of $G$ into the plane, with at most $O(\dmax \cdot
|E^*| \cdot (|E^*| + k)) = O(\dmax^{9} k^{30} \log^{19} n)$ crossings,
as required.
\end{proof}

\section{Approximating planar edge and vertex deletion}
\label{sec:planarization}
In this Section we give approximation algorithms for the minimum
planar vertex/edge deletion problems.  Recall that for a graph $G$, 
$\vertexplanarization(G)$ denotes the minimum cardinality of a set
$X\subset V(G)$, such that $G\setminus X$ is planar.  Similarly, 
$\edgeplanarization(G)$ denotes the minimum cardinality of a set
$Y\subset E(G)$, such that $G\setminus Y$ is planar. We first
derive an algorithm for $\vertexplanarization(G)$ and then 
via the degree bound assumption an algorithm for 
$\vertexplanarization(G)$.

We start with a lemma by Mohar that shows that a sufficiently large
representativitity of an embedding certifies a bound on the genus.

\begin{lemma}[Mohar \cite{Mohar_local_planarity}]\label{lem:mohar_representativity}
  Let $G$ be a graph, and let $\phi$ be a drawing of $G$ into a
  surface ${\cal S}$ of positive orientable genus, with
  representativity $\rho > 2\genus(G) + 2$.  Then, $\genus({\cal S}) =
  \genus(G)$.
\end{lemma}

The following easy lemma shows that the representativity of
an embedding can be lower bounded after the removal of a set of
vertices.

\begin{lemma}\label{lem:representativity_planarization}
  Let $G$ be a graph, and let $\phi$ be a drawing of $G$ into an
  orientable surface ${\cal S}$ of orientable genus $g>0$, with
  representativity $\rho$.  Let $X \subset V(G)$, and let $G'=G
  \setminus X$.  Let $\phi'$ be the drawing of $G'$ into ${\cal S}$
  obtained by restricting $\phi$ to $G'$, and let $\rho'$ be the
  representativity of $\phi'$.  Then, $\rho' \geq \rho - |X|$.
\end{lemma}
\begin{proof}
  Let $\gamma'$ be a noose in $\phi'$.  We can modify $\gamma'$ to
  obtain a noose $\gamma$ in $\phi$ as follows.  Consider a face $F
  \subset G'$ in $\phi'$ that is visited by $\gamma$, and suppose that
  $F$ is not a face in $\phi$.  Let $Z\subseteq X$ be the set of
  vertices in $G$ with an image in the disk ${\cal D}$ bounded by $F$.
  We shortcut the noose $\gamma'$, so that inside ${\cal D}$ it
  intersects $\phi(G)$ only on $\phi(Z)$.  This can clearly be done by
  increasing the length of the noose by at most $|Z|$.  We repeat the
  same process for all faces visited by $\gamma'$.  The resulting
  curve $\gamma$ is a noose in $\phi$, and with total length at most
  $\len(\gamma')+|X|$.  This implies that $\rho' \geq \rho - |X|$, as
  required.
\end{proof}

Now we obtain our algorithm for $\vertexplanarization(G)$ encapsulated
in the theorem below.

\begin{theorem}[Approximating the minimum planar vertex deletion]\label{thm:vertex_planarization}
  There exists a polynomial-time algorithm which given a graph $G$ of
  maximum degree $\dmax$, and an integer $k>0$, either correctly
  decides that $\vertexplanarization(G)>k$, or outputs a set
  $X\subseteq V(G)$, with $|X| = O(\dmax^4 k^{15} \log^{19/2} n)$,
  such that $G\setminus X$ is planar.
\end{theorem}

\begin{proof}
  The algorithm is similar to the one used in Theorem
  \ref{thm:crossing_number} for approximating the crossing number of a
  graph.  We may assume that $G$ is not planar, since otherwise the
  assertion is trivially true.

  We have that $\genus(G) \leq \dmax \cdot \vertexplanarization(G)$.
  Therefore, using the algorithm from Theorem
  \ref{thm:orientable_genus}, we can either correctly decide that
  $\vertexplanarization(G) > k$, or we can compute a drawing $\phi$ of
  $G$ into an orientable surface ${\cal S}$ of orientable genus
  $g=O(\dmax^3 k^{14} \log^{19/2} n)$.  We next inductively compute a
  sequence of subgraphs $G=G_0 \supset \ldots \supset G_t$.  For every
  $i\in \{0,\ldots,t\}$, we also compute a drawing $\phi_i$ of $G_i$
  into some orientable surface ${\cal S}_i$.  Set $\phi_0=\phi$, and
  ${\cal S}_0={\cal S}$.

  Given $G_i$, $\phi_i$, and ${\cal S}_i$, we proceed as follows.  If
  $G_i$ is planar, then we set $t=i$, and we terminate the sequence at
  $G_i$.  Otherwise, let $\rho_i$ be the representativity of $\phi_i$.
  We find a shortest nonseparating noose $\gamma_i$ in $\phi_i$ (using
  the algorithm from \cite{DBLP:conf/compgeom/CabelloVL10a}).
    
  Suppose first that $\len(\gamma_i) = \rho_i > (2\dmax+1)k + 2$.
  Assume that $\vertexplanarization(G) \leq k$, then
  $\vertexplanarization(G_i) \le k$ since $G_i \subset G$. Then,
  there exists $X^*_i \subset V(G_i)$, with $|X^*_i| \leq k$, and
  such that the graph $G'_i=G_i\setminus X^*_i$ is planar.  Let $\phi'_i$
  be the drawing of $G'_i$ into ${\cal S}_i$ induced by restricting $\phi_i$
  to $G'_i$, and let $\rho'_i$ be the representativity of $\phi'_i$.  By
  Lemma \ref{lem:representativity_planarization} we have
\begin{align*}
\rho'_i &\geq \rho_i - k\\
 &> 2 \cdot \dmax \cdot k + 2\\
 &> 2 \cdot \dmax \cdot \vertexplanarization(G_i) + 2\\
 &> 2 \cdot \genus(G_i)  + 2\\
 &> 2 \cdot \genus(G'_i)  + 2.
\end{align*}
Combining with Lemma \ref{lem:mohar_representativity} we deduce that
$\genus(G'_i) = \genus({\cal S}_i) > 0$, i.e.~that $G'_i$ is not
planar, which is a contradiction.  Therefore, if $\rho_i > (2\dmax+1)k
+ 2$, we obtain a proof that $\vertexplanarization(G) \ge
\vertexplanarization(G_i) > k$.

It remains to consider the case $\len(\gamma_i) = \rho_i \leq
(2\dmax+1)k + 2$.  In this case, we set $G_{i+1} = G_i \setminus X_i$,
where $X_i$ is the set of vertices in $\gamma_i$.  We also set ${\cal
  S}_{i+1}$ to be the surface obtained form ${\cal S}_i$ by cutting
along $\gamma_i$, and placing disks on any possible created punctures.
Let also $\phi_{i+1}$ be the induced drawing of $G_{i+1}$ into ${\cal
  S}_{i+1}$.  This completes the definition of the sequence
$G_0,\ldots,G_t$.  As in the proof of Theorem
\ref{thm:crossing_number}, we have $t = O(g)$.

  We set $X = \bigcup_{i=0}^{t-1} X_i$.
  We have
\begin{align*}
|X| &\leq \sum_{i=1}^t \rho_i\\
 &\leq t (2\dmax+1)k + 2\\
 &= O(\dmax^4 k^{15} \log^{19/2} n).
\end{align*}
By construction, the graph $G_t = G \setminus X$ is planar, which concludes the proof.
\end{proof}

Observe that for any graph $G$ of maximum degree $\dmax$, the
parameters $\edgeplanarization(G)$, and $\vertexplanarization(G)$ can
differ by at most a multiplicative factor of $\dmax$.  Combining this
simple observation with Theorem \ref{thm:vertex_planarization}, we
immediately obtain the following result.

\begin{theorem}[Approximating the minimum planar edge deletion]\label{thm:edge_planarization}
  There exists a polynomial-time algorithm which given a graph $G$ of
  maximum degree $\dmax$, and an integer $k>0$, either correctly
  decides that $\edgeplanarization(G)>k$, or outputs a set $Y\subseteq
  E(G)$, with $|Y| = O(\dmax^5 k^{15} \log^{19/2} n)$, such that
  $G\setminus Y$ is planar.
\end{theorem}

\fi

\myparagraph{Acknowledgments:} We thank Jeff Erickson and
Ken-ichi Kawarabayashi for useful discussions. We thank two anonymous
reviewers of our FOCS 2013 submission for pointing out technical issues
in a few proofs.

\ifabstract
\bibliographystyle{abbrv}
\fi
\iffull
\bibliographystyle{alpha}
\fi
\bibliography{bibfile}

\end{document}